\documentclass{article}

\usepackage{color}
\usepackage{amsmath,amsthm}
\usepackage{amsfonts,amssymb}
\usepackage[draft]{fixme}

\definecolor{myurlcolor}{rgb}{0.6,0,0}
\definecolor{mycitecolor}{rgb}{0,0,0.8}
\definecolor{myrefcolor}{rgb}{0,0,0.8}
\usepackage{hyperref}
\hypersetup{colorlinks,
linkcolor=myrefcolor,
citecolor=mycitecolor,
urlcolor=myurlcolor}

\usepackage[all,2cell,cmtip]{xy}
\UseTwocells
\usepackage[svgnames]{xcolor}
\usepackage{tikz}
\usetikzlibrary{backgrounds,circuits,circuits.ee.IEC,shapes,fit,matrix}

\tikzstyle{simple}=[-,line width=2.000]
\tikzstyle{arrow}=[-,postaction={decorate},decoration={markings,mark=at position .5 with {\arrow{>}}},line width=1.100]
\pgfdeclarelayer{edgelayer}
\pgfdeclarelayer{nodelayer}
\pgfsetlayers{edgelayer,nodelayer,main}

\tikzstyle{none}=[inner sep=0pt]

\definecolor{lblue}{rgb}{0,250,255}
\tikzstyle{species}=[circle,fill=yellow,draw=black,scale=1.15]
\tikzstyle{transition}=[rectangle,fill=lblue,draw=black,scale=1.15]
\tikzstyle{inarrow}=[->, >=stealth, shorten >=.03cm,line width=1.5]
\tikzstyle{empty}=[circle,fill=none, draw=none]
\tikzstyle{inputdot}=[circle,fill=purple,draw=purple, scale=.25]
\tikzstyle{inputarrow}=[->,draw=purple, shorten >=.05cm]
\tikzstyle{simple}=[-,draw=purple,line width=1.000]

\newcommand{\N}{\mathbb{N}}

\newcommand{\R}{\mathbb{R}}

\newcommand{\FinSet}{\mathtt{FinSet}}
\newcommand{\Set}{\mathtt{Set}}

\newcommand{\Cospan}{\mathtt{Cospan}}
\newcommand{\RNet}{\mathtt{RNet}}
\newcommand{\RxNet}{\mathtt{RxNet}}

\newcommand{\Dynam}{\mathtt{Dynam}}
\newcommand{\Mark}{\mathtt{Mark}}
\newcommand{\DetBalMark}{\mathtt{DetBalMark}}
\newcommand{\Rel}{\mathtt{Rel}}
\newcommand{\LinRel}{\mathtt{LinRel}}
\newcommand{\SemiAlgRel}{\mathtt{SemiAlgRel}}
\newcommand{\CC}{\mathtt{C}}
\newcommand{\D}{\mathtt{D}}
\newcommand{\Circ}{\mathtt{Circ}}

\newcommand*{\darkgraysquare}{\textcolor{gray}{\blacksquare}}
\newcommand*{\vdarkgraysquare}{\textcolor{darkgray}{\blacksquare}}
\newcommand*{\graysquare}{\textcolor{lightgray}{\blacksquare}}

\newcommand{\maps}{\colon}

\newcommand{\im}{\mathrm{im}}

\newcommand{\relto}{\nrightarrow}

\newcommand{\beq}{\begin{equation}}
\newcommand{\eeq}{\end{equation}}

\newcommand{\define}[1]{{\bf \boldmath{#1}}}

\theoremstyle{plain}

\newtheorem{thm}{Theorem}

\newtheorem{lem}[thm]{Lemma}

\newtheorem{defn}[thm]{Definition}

\theoremstyle{remark}

\begin{document}

\begin{center}   
  {\bf A Compositional Framework for Reaction Networks \\}   
  \vspace{0.3cm}
  {\em John\ C.\ Baez \\}
  \vspace{0.3cm}
  {\small
 Department of Mathematics \\
    University of California \\
  Riverside CA, USA 92521 \\ and \\
 Centre for Quantum Technologies  \\
    National University of Singapore \\
    Singapore 117543  \\    } 
  \vspace{0.4cm}
{\em Blake S. Pollard \\ }
\vspace{0.3cm}
{\small Department of Physics and Astronomy \\
University of California \\
Riverside CA 92521 \\ }
  \vspace{0.3cm}   
  {\small email:  baez@math.ucr.edu, bpoll002@ucr.edu\\} 
  \vspace{0.3cm}   
  {\small May 19, 2017}
  \vspace{0.3cm}   
\end{center}

\begin{abstract}
Reaction networks, or equivalently Petri nets, are a general framework for describing processes in which entities of various kinds interact and turn into other entities.  In chemistry, where the reactions are assigned `rate constants', any reaction network gives rise to a nonlinear dynamical system called its `rate equation'.   Here we generalize these ideas to   `open' reaction networks, which allow entities to flow in and out at certain designated inputs and outputs.  We treat open reaction networks as morphisms in a category.  Composing two such morphisms connects the outputs of the first to the inputs of the second. We construct a functor sending any open reaction network to its corresponding `open dynamical system'.  This provides a compositional framework for studying the dynamics of reaction networks.  We then turn to statics: that is, steady state solutions of open dynamical systems.   We construct a `black-boxing' functor that sends any open dynamical system to the relation that it imposes between input and output variables in steady states.  This extends our earlier work on black-boxing for Markov processes.
\end{abstract}

\section{Introduction}
\label{sec:intro}

Reaction networks, first formally defined by Aris \cite{A} in 1965, are a framework 
for describing processes whereby entities interact and transform into
other entities.  While they first arose in chemistry, and are often called `chemical
reaction networks', their applications are widespread.  For example, a basic model of infectious disease, the SIRS model, is described by this reaction network:
\[      S + I \stackrel{\iota}{\longrightarrow} 2 I  \qquad 
I \stackrel{\rho}{\longrightarrow} R \stackrel{\lambda}{\longrightarrow} S  \]
We see here three types of entity, called `species':
\begin{itemize}
\item $S$: \define{susceptible},
\item $I$: \define{infected}, 
\item $R$: \define{resistant}.
\end{itemize}
We also have three `reactions':
\begin{itemize}
\item $\iota\maps S + I \to 2 I$: \define{infection}, in which a susceptible individual meets an infected one and becomes infected;
\item $\rho\maps I \to R$: \define{recovery}, in which an infected individual gains resistance to the disease;
\item $\lambda\maps R \to S$: \define{loss of resistance}, in which a resistant individual becomes susceptible.
\end{itemize}
In general, a reaction network involves a finite set of species, but reactions go between `complexes', which are finite linear combinations of these species with natural number coefficients.  The reaction network is a directed graph whose vertices are certain complexes and whose edges are called `reactions'.   

If we attach a positive real number called a `rate constant' to each reaction, a reaction network determines a system of differential equations saying how the concentrations of the species change over time.    This system of equations is usually called 
the `rate equation'.   In the example above, the rate equation is
\beq
\label{SIRS_rate_equation}
\begin{array}{ccl}
\displaystyle{\frac{d S}{d t}} &=& r_\lambda R - r_\iota S I \\ \\
\displaystyle{\frac{d I}{d t}} &=&  r_\iota S I - r_\rho I \\  \\
\displaystyle{\frac{d R}{d t}} &=& r_\rho I - r_\lambda R .
\end{array}
\eeq
Here $r_\iota, r_\rho$ and $r_\lambda$ are the rate constants for the three
reactions, and $S, I, R$ now stand for the concentrations of the three species, which are treated in a continuum approximation as smooth functions $S, I, R \maps \R \to [0,\infty)$.  The rate equation can be derived from the `law of mass action', which says that any reaction occurs at a rate equal to its rate constant times the product of the concentrations of the species entering it as inputs.    

A reaction network is more than just a stepping-stone to its rate equation.  Interesting qualitative properties of the rate equation, such as existence and uniqueness of steady state solutions, can often be determined simply by examining the reaction network, independent of any particular choice of rate constants.  Results in this direction began with  Feinberg and Horn's seminal work in the 1960's \cite{Feinberg,FeinbergHorn}, leading to the Deficiency Zero and Deficiency One Theorems \cite{Feinberg1995a,Feinberg1995b} and more recently to a proof of the Global Attractor Conjecture \cite{Craciun}.     

In this paper we present a `compositional framework' for reaction networks: that is, a way to build up a reaction network from smaller pieces, in such a way that its rate equation can be determined from those of the pieces.   However, this framework requires that we view reaction networks in a somewhat different way, as `Petri nets'.

Petri nets were invented by Carl Petri in 1939, when he was just a teenager, for the purposes of chemistry \cite{PR}.  Much later, they became popular in theoretical computer science \cite{JK,Peterson}, biology \cite{K,KRS,Wilkinson}, and other fields \cite{BaezBiamonte,Haas,MBCDF}. A Petri net is a bipartite directed graph; vertices of one kind represent species, while those of the other kind represent reactions.  The edges into a reaction specify which species are inputs to that reaction, while the edges out specify its outputs.   One can easily turn a reaction network into a Petri net and vice versa.    For example, the reaction network above translates into this Petri net:
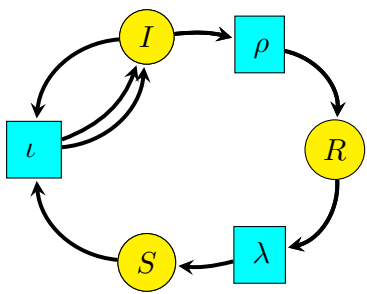
\[
\begin{tikzpicture}
	\begin{pgfonlayer}{nodelayer}
		\node [style=species] (S) at (0, -1.5) {$S$};
		\node [style=species] (I) at (0, 1.5) {$I$};
		\node [style=transition] (i) at (-1.5,0) {$\; \big. \iota$ \,};   
		\node [style=none] (ATL) at (-3, 2) {};
		\node [style=none] (ATR) at (-2.4, 2) {};
		\node [style=none] (ABR) at (-2.4, -2) {};
		\node [style=none] (ABL) at (-3, -2) {};
		\node [style=none] (BTL) at (0.7, 2) {};
		\node [style=none] (BTR) at (1.3, 2) {};
		\node [style=none] (BBR) at (1.3, -2) {};
		\node [style=none] (BBL) at (0.7, -2) {};
	      \node [style=species] (S) at (0, -1.5) {$S$};
		\node [style=species] (I) at (0, 1.5) {$I$};
             \node [style=species] (R) at (2.5,0) {$R$};
             \node [style=transition] (r) at (1.5,1.4) {$\; \big. \rho\,$};    
		\node [style=transition] (l) at (1.5,-1.4) {$\; \big. \lambda \,$};    
		\node [style=none] (A'TL) at (-1.3, 2) {};
		\node [style=none] (A'TR) at (-0.7, 2) {};
		\node [style=none] (A'BR) at (-0.7, -2) {};
		\node [style=none] (A'BL) at (-1.3, -2) {};
		\node [style=none] (B'TL) at (3.2, 2) {};
		\node [style=none] (B'TR) at (3.8, 2) {};
		\node [style=none] (B'BR) at (3.8, -2) {};
		\node [style=none] (B'BL) at (3.2, -2) {};
		
	\end{pgfonlayer}
	\begin{pgfonlayer}{edgelayer}
		\draw [style=inarrow, bend left=40, looseness=1.00] (S) to (i);
		\draw [style=inarrow, bend right=40, looseness=1.00] (I) to (i);
		\draw [style=inarrow, bend right=25, looseness=1.00] (i) to (I);
		\draw [style=inarrow, bend right=40, looseness=1.00] (i) to (I);
	       \draw [style=inarrow, bend left=40, looseness=1.00] (R) to (l);
		\draw [style=inarrow, bend left=40, looseness=1.00] (r) to (R);
		\draw [style=inarrow,  bend left=10, looseness=1.00] (I) to (r);
	      \draw [style=inarrow, bend left=40, looseness=1.00] (R) to (l);
		\draw [style=inarrow,  bend left=10, looseness=1.00] (l) to (S);
		\draw [style=inarrow, bend left=40, looseness=1.00] (r) to (R);
		\draw [style=inarrow,  bend left=10, looseness=1.00] (I) to (r);
	\end{pgfonlayer}
\end{tikzpicture}
\]
One should beware that the terminology is diverse, since it comes from several communities.  In the Petri net literature, species are called `places' and reactions are called `transitions'.  Indeed, Petri nets are sometimes called `place-transition nets' or `P/T nets'.  On the other hand, chemists call them `species-reaction graphs' or `SR-graphs'.   When each reaction of a Petri net has a rate constant attached to it, it is often called a `stochastic Petri net'.

While some qualitative properties of a rate equation can be read off from a reaction network, others are more easily read from the corresponding Petri net.  For example, properties of a Petri net can be used to determine whether its rate equation has the capacity to admit multiple steady states \cite{BanajiCraciun,CraciunFeinbergTang, FeinbergShinar}.  

Petri nets are also better suited to a compositional framework.  The key new concept
required is that of an `open' Petri net.  Here is an example:
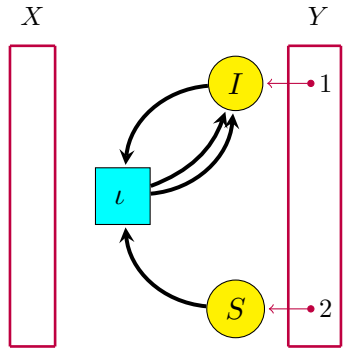
\[
\begin{tikzpicture}
	\begin{pgfonlayer}{nodelayer}
		\node [style=species] (S) at (0, -1.5) {$S$};
		\node [style=species] (I) at (0, 1.5) {$I$};
		\node [style=transition] (i) at (-1.5,0) {$\; \big. \iota$ \,};
		\node [style=inputdot] (outI) at (1, 1.5) {};
		\node [style=inputdot] (outS) at (1, -1.5) {};
		\node [style=none] (ATL) at (-3, 2) {};
		\node [style=none] (ATR) at (-2.4, 2) {};
		\node [style=none] (ABR) at (-2.4, -2) {};
		\node [style=none] (ABL) at (-3, -2) {};
		\node [style=none] (BTL) at (0.7, 2) {};
		\node [style=none] (BTR) at (1.4, 2) {};
		\node [style=none] (BBR) at (1.4, -2) {};
		\node [style=none] (BBL) at (0.7, -2) {};
		\node [style=empty] at (-2.7, 2.4) {$X$};
		\node [style=empty] at (1.1, 2.4) {$Y$};
		\node [style=empty] at (1.2, 1.5) {$1$}; 
		\node [style=empty] at (1.2, -1.5) {$2$}; 
		
	\end{pgfonlayer}
	\begin{pgfonlayer}{edgelayer}
		\draw [style=inarrow, bend left=40, looseness=1.00] (S) to (i);
		\draw [style=inarrow, bend right=40, looseness=1.00] (I) to (i);
		\draw [style=inarrow, bend right=25, looseness=1.00] (i) to (I);
		\draw [style=inarrow, bend right=40, looseness=1.00] (i) to (I);
		\draw [style=inputarrow] (outI) to (I);
		\draw [style=inputarrow] (outS) to (S);
		\draw [style=simple] (ATL.center) to (ATR.center);
		\draw [style=simple] (ATR.center) to (ABR.center);
		\draw [style=simple] (ABR.center) to (ABL.center);
		\draw [style=simple] (ABL.center) to (ATL.center);
		\draw [style=simple] (BTL.center) to (BTR.center);
		\draw [style=simple] (BTR.center) to (BBR.center);
		\draw [style=simple] (BBR.center) to (BBL.center);
		\draw [style=simple] (BBL.center) to (BTL.center);

	\end{pgfonlayer}
\end{tikzpicture}
\]
The box at left shows a set $X$ of `inputs' (which happens to be empty), while the box at right shows a set $Y$ of `outputs'.  Both inputs and outputs are points at which entities of various species can flow in or out of the Petri net.  We say the open Petri net goes from $X$ to $Y$, and we shall show how to treat it as a morphism $f \maps X \to Y$ in a category we call $\RxNet$.     

Given an open Petri net with rate constants assigned to each reaction, we explain how to systematically obtain its `open rate equation', which amounts to the usual rate equation with extra terms describing inflows and outflows.  The above example has this open rate equation:
\beq
\begin{array}{ccr}
\label{open_rate_equation_1}
\displaystyle{\frac{d S}{d t}} &=&  - r_\iota S I - o_1 \\ \\
\displaystyle{\frac{d I}{d t}} &=&  r_\iota S I - o_2 .
\end{array}
\eeq
Here $o_1, o_2 \maps \R \to \R$ are arbitrary smooth functions describing outflows as a function of time.  

Given another open Petri net $g \maps Y \to Z$, for example this:
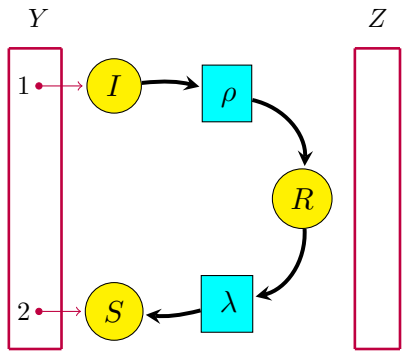
\[
\begin{tikzpicture}
	\begin{pgfonlayer}{nodelayer}
		\node [style=species] (S) at (0, -1.5) {$S$};
		\node [style=species] (I) at (0, 1.5) {$I$};
             \node [style=species] (R) at (2.5,0) {$R$};
             \node [style=transition] (r) at (1.5,1.4) {$\; \big. \rho\,$};    
		\node [style=transition] (l) at (1.5,-1.4) {$\; \big. \lambda \,$};    
		\node [style=inputdot] (inI) at (-1, 1.5) {};
		\node [style=inputdot] (inS) at (-1, -1.5) {};
		\node [style=none] (A'TL) at (-1.4, 2) {};
		\node [style=none] (A'TR) at (-0.7, 2) {};
		\node [style=none] (A'BR) at (-0.7, -2) {};
		\node [style=none] (A'BL) at (-1.4, -2) {};
		\node [style=none] (B'TL) at (3.2, 2) {};
		\node [style=none] (B'TR) at (3.8, 2) {};
		\node [style=none] (B'BR) at (3.8, -2) {};
		\node [style=none] (B'BL) at (3.2, -2) {};
		\node [style=empty] at (-1, 2.4) {$Y$};
		\node [style=empty] at (3.5, 2.4) {$Z$};
		\node [style=empty] at (-1.2, 1.5) {$1$}; 
		\node [style=empty] at (-1.2, -1.5) {$2$}; 
		
	\end{pgfonlayer}
	\begin{pgfonlayer}{edgelayer}
		\draw [style=inarrow, bend left=40, looseness=1.00] (R) to (l);
		\draw [style=inarrow,  bend left=10, looseness=1.00] (l) to (S);
		\draw [style=inarrow, bend left=40, looseness=1.00] (r) to (R);
		\draw [style=inarrow,  bend left=10, looseness=1.00] (I) to (r);
		\draw [style=inputarrow] (inI) to (I);
		\draw [style=inputarrow] (inS) to (S);
		\draw [style=simple] (A'TL.center) to (A'TR.center);
		\draw [style=simple] (A'TR.center) to (A'BR.center);
		\draw [style=simple] (A'BR.center) to (A'BL.center);
		\draw [style=simple] (A'BL.center) to (A'TL.center);
		\draw [style=simple] (B'TL.center) to (B'TR.center);
		\draw [style=simple] (B'TR.center) to (B'BR.center);
		\draw [style=simple] (B'BR.center) to (B'BL.center);
		\draw [style=simple] (B'BL.center) to (B'TL.center);

	\end{pgfonlayer}
\end{tikzpicture}
\]
it will have its own open rate equation, in this case
\beq
\label{open_rate_equation_2}
\begin{array}{ccc}
\displaystyle{\frac{d S}{d t}} &=& r_\lambda R + i_2 \\ \\
\displaystyle{\frac{d I}{d t}} &=& - r_\rho I + i_1 \\  \\
\displaystyle{\frac{d R}{d t}} &=& r_\rho I - r_\lambda R .
\end{array}
\eeq
Here $i_1, i_2 \maps \R \to \R$ are arbitrary smooth functions describing inflows.  We can compose $f$ and $g$ by gluing the outputs of $f$ to the inputs of $g$, obtaining a new open Petri net $gf \maps X \to Z$, as follows:
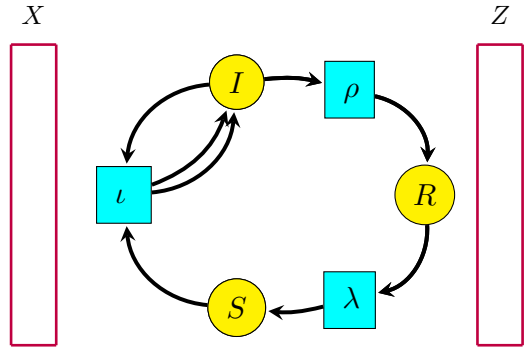
\[
\begin{tikzpicture}
	\begin{pgfonlayer}{nodelayer}
		\node [style=species] (S) at (0, -1.5) {$S$};
		\node [style=species] (I) at (0, 1.5) {$I$};
		\node [style=transition] (i) at (-1.5,0) {$\; \big. \iota$ \,};   
		\node [style=none] (ATL) at (-3, 2) {};
		\node [style=none] (ATR) at (-2.4, 2) {};
		\node [style=none] (ABR) at (-2.4, -2) {};
		\node [style=none] (ABL) at (-3, -2) {};
		\node [style=none] (BTL) at (0.7, 2) {};
		\node [style=none] (BTR) at (1.3, 2) {};
		\node [style=none] (BBR) at (1.3, -2) {};
		\node [style=none] (BBL) at (0.7, -2) {};
		\node [style=empty] at (-2.7, 2.4) {$X$};
	      \node [style=species] (S) at (0, -1.5) {$S$};
		\node [style=species] (I) at (0, 1.5) {$I$};
             \node [style=species] (R) at (2.5,0) {$R$};
             \node [style=transition] (r) at (1.5,1.4) {$\; \big. \rho\,$};    
		\node [style=transition] (l) at (1.5,-1.4) {$\; \big. \lambda \,$};    
		\node [style=none] (A'TL) at (-1.3, 2) {};
		\node [style=none] (A'TR) at (-0.7, 2) {};
		\node [style=none] (A'BR) at (-0.7, -2) {};
		\node [style=none] (A'BL) at (-1.3, -2) {};
		\node [style=none] (B'TL) at (3.2, 2) {};
		\node [style=none] (B'TR) at (3.8, 2) {};
		\node [style=none] (B'BR) at (3.8, -2) {};
		\node [style=none] (B'BL) at (3.2, -2) {};
		\node [style=empty] at (3.5, 2.4) {$Z$};
		
	\end{pgfonlayer}
	\begin{pgfonlayer}{edgelayer}
		\draw [style=inarrow, bend left=40, looseness=1.00] (S) to (i);
		\draw [style=inarrow, bend right=40, looseness=1.00] (I) to (i);
		\draw [style=inarrow, bend right=25, looseness=1.00] (i) to (I);
		\draw [style=inarrow, bend right=40, looseness=1.00] (i) to (I);
		\draw [style=simple] (ATL.center) to (ATR.center);
		\draw [style=simple] (ATR.center) to (ABR.center);
		\draw [style=simple] (ABR.center) to (ABL.center);
		\draw [style=simple] (ABL.center) to (ATL.center);
		\draw [style=simple] (B'TL.center) to (B'TR.center);
		\draw [style=simple] (B'TR.center) to (B'BR.center);
		\draw [style=simple] (B'BR.center) to (B'BL.center);
		\draw [style=simple] (B'BL.center) to (B'TL.center);
	       \draw [style=inarrow, bend left=40, looseness=1.00] (R) to (l);
		\draw [style=inarrow, bend left=40, looseness=1.00] (r) to (R);
		\draw [style=inarrow,  bend left=10, looseness=1.00] (I) to (r);
	      \draw [style=inarrow, bend left=40, looseness=1.00] (R) to (l);
		\draw [style=inarrow,  bend left=10, looseness=1.00] (l) to (S);
		\draw [style=inarrow, bend left=40, looseness=1.00] (r) to (R);
		\draw [style=inarrow,  bend left=10, looseness=1.00] (I) to (r);
	\end{pgfonlayer}
\end{tikzpicture}
\]
Since this open Petri net $gf$ has no inputs or outputs, it amounts to an ordinary Petri net, and its open rate equation is a rate equation of the usual kind.  Indeed, this is the Petri net we have already seen, and its open rate equation is Equation \eqref{SIRS_rate_equation}.  

There is a systematic procedure for combining the open rate equations for
two open Petri nets to obtain that of their composite.  In the case at hand it amounts to identifying the outflows of $f$ with the inflows of $g$ (setting $i_1 = o_1$ and $i_2 = o_2$) and then adding the right hand sides of Equations \eqref{open_rate_equation_1} and \eqref{open_rate_equation_2}.  The first goal of this paper is to precisely describe this procedure, and to prove that it defines a functor
\[     \graysquare\maps \RxNet \to \Dynam \]
from $\RxNet$ to a category $\Dynam$ where the morphisms are `open dynamical systems'.  By a dynamical system, we essentially mean a vector field on $\R^n$, which can be used to define a system of first-order ordinary differential equations in $n$ variables.  An example is the rate equation of a Petri net.  An open dynamical system allows for the possibility of extra terms that are arbitrary functions of time, such as the inflows and outflows in an open rate equation.

In fact, we prove that $\RxNet$ and $\Dynam$ are symmetric monoidal categories and that $\graysquare$ is a symmetric monoidal functor.  To do this, we use the machinery of `decorated cospans', developed by Fong \cite{Fong2015,FongThesis} and already applied to electrical circuits \cite{BaezFong} and Markov processes \cite{BaezFongPollard}.   Decorated cospans provide a powerful general tool for describing open systems.   A cospan in any category is a diagram of the form
\[ \xymatrix{ & S & \\ X \ar[ur]^{i} & & Y. \ar[ul]_{o} } \]
We are mostly interested in cospans in $\FinSet$, the category of finite sets and functions between these.  The set $S$, the \define{apex} of the cospan, is the set of states of an open system.  The sets $X$ and $Y$ are the \define{inputs} and \define{outputs} of this system.  The \define{legs} of the cospan, meaning the morphisms $i \maps X \to S$ and $o \maps Y \to S$, describe how these inputs and outputs are included in the system.  In our application, $S$ is the set of species of a Petri net.  

For example, we may take this reaction network: 
\[ A+B \stackrel{\alpha}{\longrightarrow} 2C \quad \quad C \stackrel{\beta}{\longrightarrow} D, \]
treat it as a Petri net with $S = \{A,B,C,D\}$:
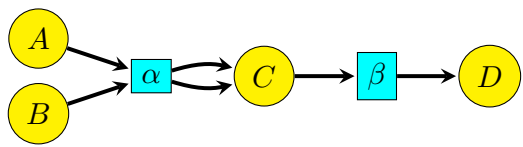
\[
\begin{tikzpicture}
	\begin{pgfonlayer}{nodelayer}
		\node [style=transition] (-1) at (0.5, 0) {$\beta$};
		\node [style=species] (-2) at (2, 0) {$D$};
		\node [style=species] (0) at (-1, 0) {$C$};
		\node [style=none] (1) at (-0.25, 0.75) {};
		\node [style=species] (2) at (-4, -0.5) {$B$};
		\node [style=none] (3) at (-4.75, 0.75) {};
		\node [style=none] (4) at (-5.25, 0.75) {};
		\node [style=none] (6) at (0.25, 0.75) {};
		\node [style=transition] (7) at (-2.5, 0) {$\alpha$};
		\node [style=none] (8) at (-5.25, -0.75) {};
		\node [style=none] (11) at (-0.25, -0.75) {};
		\node [style=none] (14) at (-4.75, -0.75) {};
		\node [style=none] (15) at (0.25, -0.75) {};
		\node [style=species] (18) at (-4, 0.5) {$A$};
	\end{pgfonlayer}
	\begin{pgfonlayer}{edgelayer}
		\draw [style=inarrow] (-1) to (-2);
		\draw [style=inarrow] (0) to (-1);
		\draw [style=inarrow] (18) to (7);
		\draw [style=inarrow] (2) to (7);
		\draw [style=inarrow, bend right=15, looseness=1.00] (7) to (0);
		\draw [style=inarrow, bend left=15, looseness=1.00] (7) to (0);
	\end{pgfonlayer}
\end{tikzpicture}
\]
and then turn that into an open Petri net by choosing any finite sets $X,Y$ and
maps $i \maps X \to S$, $o \maps Y \to S$, for example as follows:
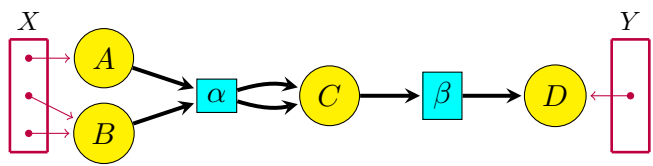
\[
\begin{tikzpicture}
	\begin{pgfonlayer}{nodelayer}
		\node [style=transition] (-1) at (0.5, 0) {$\beta$};
		\node [style=species] (-2) at (2, 0) {$D$};
		\node [style=species] (0) at (-1, 0) {$C$};
		\node [style=none] (1) at (2.75, 0.75) {};
		\node [style=species] (2) at (-4, -0.5) {$B$};
		\node [style=none] (3) at (-4.75, 0.75) {};
		\node [style=none] (4) at (-5.25, 0.75) {};
		\node [style=inputdot] (5) at (3, 0) {};
		\node [style=none] (6) at (3.25, 0.75) {};
		\node [style=transition] (7) at (-2.5, 0) {$\alpha$};
		\node [style=none] (8) at (-5.25, -0.75) {};
		\node [style=empty] (10) at (-5, 1) {$X$};
		\node [style=none] (11) at (2.75, -0.75) {};
		\node [style=inputdot] (12) at (-5, -0.5) {};
		\node [style=inputdot] (13) at (-5, 0.5) {};
		\node [style=none] (14) at (-4.75, -0.75) {};
		\node [style=none] (15) at (3.25, -0.75) {};
		\node [style=inputdot] (16) at (-5, 0) {};
		\node [style=empty] (17) at (3, 1) {$Y$};
		\node [style=species] (18) at (-4, 0.5) {$A$};
	\end{pgfonlayer}
	\begin{pgfonlayer}{edgelayer}
		\draw [style=inarrow] (-1) to (-2);
		\draw [style=inarrow] (0) to (-1);
		\draw [style=inarrow] (18) to (7);
		\draw [style=inarrow] (2) to (7);
		\draw [style=inarrow, bend right=15, looseness=1.00] (7) to (0);
		\draw [style=inarrow, bend left=15, looseness=1.00] (7) to (0);
		\draw [style=inputarrow] (5) to (-2);
		\draw [style=inputarrow] (13) to (18);
		\draw [style=inputarrow] (16) to (2);
		\draw [style=inputarrow] (12) to (2);
		\draw [style=simple] (4.center) to (3.center);
		\draw [style=simple] (3.center) to (14.center);
		\draw [style=simple] (14.center) to (8.center);
		\draw [style=simple] (8.center) to (4.center);
		\draw [style=simple] (1.center) to (6.center);
		\draw [style=simple] (6.center) to (15.center);
		\draw [style=simple] (15.center) to (11.center);
		\draw [style=simple] (11.center) to (1.center);
	\end{pgfonlayer}
\end{tikzpicture}
\]
Notice that the maps including the inputs and outputs into the states of the system need not be one-to-one.  This is technically useful, but it introduces some subtleties not yet explained in our discussion of the gray-boxing functor.

An open Petri net can thus be seen as a cospan of finite sets whose apex $S$ is `decorated' with some extra information, namely a Petri net with $S$ as its set of species.  Fong's theory of decorated cospans lets us define a category with open Petri nets as morphisms, with composition given by gluing the outputs of one open Petri net to the inputs of the other.   For details see Appendix \ref{sec:deccospan}, where we recall some results on constructing decorated cospan categories and functors between these.

We call the functor $\graysquare \maps \RxNet \to \Dynam$ `gray-boxing' because it hides some but not all the internal details of an open Petri net.  We can go further and `black-box' an open dynamical system.  This amounts to recording only the relation between input and output variables that must hold in steady state.  We prove that black-boxing gives a functor 
\[        \blacksquare \maps \Dynam \to \SemiAlgRel \] 
to a category $\SemiAlgRel$ of `semi-algebraic relations' between real vector spaces, meaning relations defined by polynomials and inequalities.  (This relies on the fact that our dynamical systems involve \emph{algebraic} vector fields, meaning those whose components are polynomials; more general dynamical systems would give more general relations.)  We prove that like the gray-boxing functor, the black-boxing functor is symmetric monoidal.  

This paper is structured as follows. In Section \ref{sec:petri} we review and compare reaction networks and Petri nets.  In Section \ref{sec:RNet} we construct a category $\RNet$ where an object is a finite set and a morphism is an open reaction network (or more precisely, an isomorphism class of open reaction networks).   In Section \ref{sec:RxNet} we enhance this construction to define a category $\RxNet$ where the transitions of the open reaction networks are equipped with rate constants.   In Section \ref{sec:openrate} we explain the open dynamical system associated to an open reaction network, and in Section \ref{sec:opendynam} we construct a category $\Dynam$ of open dynamical systems. In Section \ref{sec:gray} we construct the gray-boxing functor $\graysquare \maps \RxNet \to \Dynam$.   In Section \ref{sec:black} we construct the black-boxing functor $\blacksquare \maps \Dynam \to \SemiAlgRel$.  Finally, in Section \ref{sec:conclusions} we fit our results into a larger `network of network theories'.

\subsection*{Acknowledgements}

We thank Brendan Fong for many conversations, and David Spivak for pointing out his closely related work on black-boxing open dynamical systems \cite{Spivak}. We thank Daniel Cicala for helping us understand the functor $D$ sending any finite set $S$ to the set of algebraic vector fields on $\R^S$, as discussed in Section \ref{sec:opendynam}.  We also thank the Centre for Quantum Technologies and DARPA's Complex Adaptive System Composition and Design Environment program, who helped support this research.

\section{Reaction networks versus Petri nets}
\label{sec:petri}

We begin by precisely defining reaction networks and Petri nets, so we can see that they are
two ways of presenting the same concept.

\begin{defn}
A \define{reaction network} $(S,T,s,t)$ consists of:
\begin{itemize}
\item a finite set $S$,
\item a finite set $T$,
\item functions $s,t \maps T \to \N^S$.
\end{itemize}
We call the elements of $S$ \define{species}, those of $\N^S$ \define{complexes},
and those of $T$ \define{transitions}.  Any transition $\tau \in T$ has a \define{source}
$s(\tau)$ and \define{target} $t(\tau)$.  If $s(\tau) = \kappa$ and $t(\tau) = \kappa'$ 
we write $\tau \maps \kappa \to \kappa'$. 
\end{defn}

\noindent Our terminology here is a compromise: as mentioned before, in the
Petri net literature `species' are often called `places', while in the chemistry
literature `transitions' are often called `reactions'.  The alphabetical proximity
of $S$ and $T$ gives our compromise a certain charm.

The set of complexes relevant to a given reaction network $(S,T,s,t)$ is
\[            K = \im(s) \cup \im(t) \subseteq \N^S. \]
The reaction network gives a graph with $K$ as its set of vertices and a directed 
edge from $\kappa \in K$ to $\kappa' \in K$ for each transition $\tau \maps \kappa 
\to \kappa'$.     This graph may have multiple edges or self-loops.  It is thus
the kind of graph sometimes called a `directed multigraph' or
`quiver'.   However, a graph of this kind can only arise from a reaction network if 
every vertex is the source or target of some edge.  

On the other hand, we have Petri nets:

\begin{defn}
A \define{Petri net} $(S,T,m,n)$ consists of:
\begin{itemize}
\item a finite set $S$,
\item a finite set $T$,
\item functions $m,n \maps S \times T \to \N$.
\end{itemize}
We call the elements
of $S$ \define{species} and those of $T$ \define{transitions}.  
We call any Petri net with $S$ as its set of species a \define{Petri net on} $S$. 
Given a species $\sigma \in S$ and a transition $\tau \in T$, we say
$m(\sigma,\tau) \in \N$ is the number of times $\sigma$ appears
as an \define{input} of $\tau$, and $n(\sigma,\tau) \in \N$ is the 
number of times $\sigma$ appears as the \define{output} of $\tau$.
\end{defn}

It is easy to convert a reaction network into a Petri net, or vice versa, using
the relations
\[          m(\sigma,\tau) = s(\tau)(\sigma) , \]
\[          n(\sigma,\tau) = t(\tau)(\sigma)  .\]
Thus, anything we do with reaction networks we can do with Petri nets, and
vice versa.  However, Petri nets are usually drawn in a different way.   Given a Petri net $(S,T,m,n)$ we draw a yellow circle for each species $\sigma \in S$, an aqua square for each transition $\tau \in T$, a certain number $m(\sigma,\tau)$ of edges from each species $\sigma$ to each transition $\tau$, and a certain number $n(\sigma,\tau)$ of edges from each transition $\tau$ to each species $\sigma$.  This gives a picture of a Petri net as a bipartite graph.  For example, this reaction network
\[ 3A \stackrel{\tau}{\longrightarrow} B + C  \]
gives a Petri net that we can draw as follows:
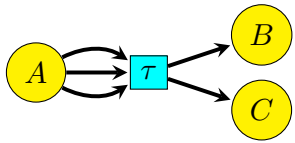
\[
\begin{tikzpicture}
	\begin{pgfonlayer}{nodelayer}
		\node [style=species] (0) at (-4, 0) {$A$};
		\node [style=none] (1) at (-0.25, 0.75) {};
		\node [style=species] (2) at (-1, -0.5) {$C$};
		\node [style=none] (3) at (-4.75, 0.75) {};
		\node [style=none] (4) at (-5.25, 0.75) {};
		\node [style=none] (6) at (0.25, 0.75) {};
		\node [style=transition] (7) at (-2.5, 0) {$\tau$};
		\node [style=none] (8) at (-5.25, -0.75) {};
		\node [style=none] (11) at (-0.25, -0.75) {};
		\node [style=none] (14) at (-4.75, -0.75) {};
		\node [style=none] (15) at (0.25, -0.75) {};
		\node [style=species] (18) at (-1, 0.5) {$B$};
	\end{pgfonlayer}
	\begin{pgfonlayer}{edgelayer}
		\draw [style=inarrow] (7) to (18);
		\draw [style=inarrow] (7) to (2);
		\draw [style=inarrow, bend right=30, looseness=1.00] (0) to (7);
		\draw [style=inarrow] (0) to (7);
		\draw [style=inarrow, bend left=30, looseness=1.00] (0) to (7);
	\end{pgfonlayer}
\end{tikzpicture}
\]
However, we must be a bit careful.   A Petri net does not have a well-defined \emph{set} of edges from a species $\sigma$ to a transition $\tau$, or vice versa: there is merely a \emph{natural number} $m(\sigma,\tau)$ and a natural number $n(\sigma,\tau)$.  Thus, unlike in a reaction network, the edges in a Petri net do not have `names'.  For example, in the picture above, we are not allowed to distinguish between the three edges from $A$ to $\tau$: these edges are just a convenient way of indicating that $m(A,\tau) = 3$.

In what follows we work with reaction networks, since our applications are especially relevant to chemistry, but we often draw them as Petri nets. 

\section{Open reaction networks}
\label{sec:RNet}

To treat a reaction network as an open system, we equip it with `inputs' and
`outputs':

\begin{defn}
\label{defn:open_reaction_network}
Given finite sets $X$ and $Y$, an \define{open reaction network from} $X$ \define{to} $Y$ is a cospan of finite sets
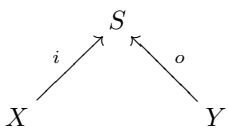
\[ \xymatrix{  & S  &  \\ X \ar[ur]^{i} & & Y \ar[ul]_{o} } \] 
together with a reaction network $R$ on $S$. We often abbreviate all this data as $R \maps X \to Y$. We say $X$ is the set of \define{inputs} and $Y$ is the set of \define{outputs}.    
\end{defn}

The first thing we want to do with open reaction networks is compose them: that is, glue them together to build larger ones.   Given open reaction networks $R \maps X \to Y$ and $R' \maps Y \to Z$, we can compose them to get an open reaction network $R' R \maps X \to Z$.  

The process of composition actually involves two steps.  These are easy to visualize if we draw each reaction network as a Petri net as explained in Section \ref{sec:petri}, and also indicate the maps from inputs and outputs into the set of species.  For example, suppose $R \maps X \to Y$ looks like this:
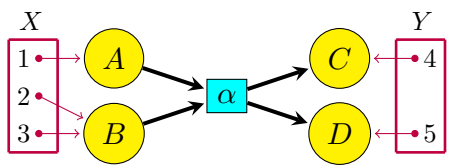
\[
\begin{tikzpicture}
	\begin{pgfonlayer}{nodelayer}
		\node [style=species] (A) at (-4, 0.5) {$A$};
		\node [style=species] (B) at (-4, -0.5) {$B$};
		\node [style=species] (C) at (-1, 0.5) {$C$};
		\node [style=species] (D) at (-1, -0.5) {$D$};
             \node [style=transition] (a) at (-2.5, 0) {$\alpha$}; 
		
		\node [style=empty] (X) at (-5.1, 1) {$X$};
		\node [style=none] (Xtr) at (-4.75, 0.75) {};
		\node [style=none] (Xbr) at (-4.75, -0.75) {};
		\node [style=none] (Xtl) at (-5.4, 0.75) {};
             \node [style=none] (Xbl) at (-5.4, -0.75) {};
	
		\node [style=inputdot] (1) at (-5, 0.5) {};
		\node [style=empty] at (-5.2, 0.5) {$1$};
		\node [style=inputdot] (2) at (-5, 0) {};
		\node [style=empty] at (-5.2, 0) {$2$};
		\node [style=inputdot] (3) at (-5, -0.5) {};
		\node [style=empty] at (-5.2, -0.5) {$3$};

		\node [style=empty] (Y) at (0.1, 1) {$Y$};
		\node [style=none] (Ytr) at (.4, 0.75) {};
		\node [style=none] (Ytl) at (-.25, 0.75) {};
		\node [style=none] (Ybr) at (.4, -0.75) {};
		\node [style=none] (Ybl) at (-.25, -0.75) {};

		\node [style=inputdot] (4) at (0, 0.5) {};
		\node [style=empty] at (0.2, 0.5) {$4$};
		\node [style=inputdot] (5) at (0, -0.5) {};
		\node [style=empty] at (0.2, -0.5) {$5$};		
		
		
	\end{pgfonlayer}
	\begin{pgfonlayer}{edgelayer}
		\draw [style=inarrow] (A) to (a);
		\draw [style=inarrow] (B) to (a);
		\draw [style=inarrow] (a) to (C);
		\draw [style=inarrow] (a) to (D);
		\draw [style=inputarrow] (1) to (A);
		\draw [style=inputarrow] (2) to (B);
		\draw [style=inputarrow] (3) to (B);
		\draw [style=inputarrow] (4) to (C);
		\draw [style=inputarrow] (5) to (D);
		\draw [style=simple] (Xtl.center) to (Xtr.center);
		\draw [style=simple] (Xtr.center) to (Xbr.center);
		\draw [style=simple] (Xbr.center) to (Xbl.center);
		\draw [style=simple] (Xbl.center) to (Xtl.center);
		\draw [style=simple] (Ytl.center) to (Ytr.center);
		\draw [style=simple] (Ytr.center) to (Ybr.center);
		\draw [style=simple] (Ybr.center) to (Ybl.center);
		\draw [style=simple] (Ybl.center) to (Ytl.center);
	\end{pgfonlayer}
\end{tikzpicture}
\]
and $R' \maps Y \to Z$ looks like this:
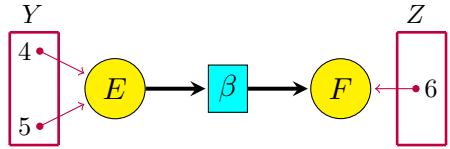
\[
\begin{tikzpicture}
	\begin{pgfonlayer}{nodelayer}

		\node [style = species] (E) at (1, 0) {$E$};
		\node [style = transition] (b) at (2.5, 0) {$\beta$};
		\node [style = species] (F) at (4,0) {$F$};
		
	

		\node [style=empty] (Y) at (-0.1, 1) {$Y$};
		\node [style=none] (Ytr) at (.25, 0.75) {};
		\node [style=none] (Ytl) at (-.4, 0.75) {};
		\node [style=none] (Ybr) at (.25, -0.75) {};
		\node [style=none] (Ybl) at (-.4, -0.75) {};

		\node [style=inputdot] (4) at (0, 0.5) {};
		\node [style=empty] at (-0.2, 0.5) {$4$};
		\node [style=inputdot] (5) at (0, -0.5) {};
		\node [style=empty] at (-0.2, -0.5) {$5$};		
		
		\node [style=empty] (Z) at (5, 1) {$Z$};
		\node [style=none] (Ztr) at (4.75, 0.75) {};
		\node [style=none] (Ztl) at (5.4, 0.75) {};
		\node [style=none] (Zbl) at (5.4, -0.75) {};
		\node [style=none] (Zbr) at (4.75, -0.75) {};

		\node [style=inputdot] (6) at (5, 0) {};
		\node [style=empty] at (5.2, 0) {$6$};	
		
	\end{pgfonlayer}
	\begin{pgfonlayer}{edgelayer}
		\draw [style=inarrow] (E) to (b);
		\draw [style=inarrow] (b) to (F);
		\draw [style=inputarrow] (4) to (E);
		\draw [style=inputarrow] (5) to (E);
		\draw [style=inputarrow] (6) to (F);
		\draw [style=simple] (Ytl.center) to (Ytr.center);
		\draw [style=simple] (Ytr.center) to (Ybr.center);
		\draw [style=simple] (Ybr.center) to (Ybl.center);
		\draw [style=simple] (Ybl.center) to (Ytl.center);
		\draw [style=simple] (Ztl.center) to (Ztr.center);
		\draw [style=simple] (Ztr.center) to (Zbr.center);
		\draw [style=simple] (Zbr.center) to (Zbl.center);
		\draw [style=simple] (Zbl.center) to (Ztl.center);
	\end{pgfonlayer}
\end{tikzpicture}
\]
To compose them, the first step is to put the pictures together:
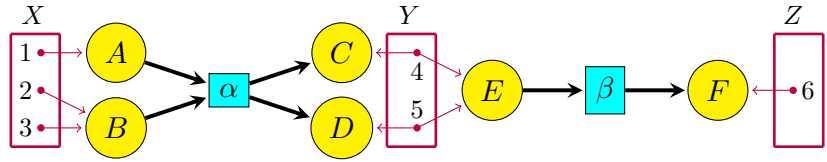
\[
\begin{tikzpicture}
	\begin{pgfonlayer}{nodelayer}
		\node [style=species] (A) at (-4, 0.5) {$A$};
		\node [style=species] (B) at (-4, -0.5) {$B$};
		\node [style=species] (C) at (-1, 0.5) {$C$};
		\node [style=species] (D) at (-1, -0.5) {$D$};
             \node [style=transition] (a) at (-2.5, 0) {$\alpha$}; 
		\node [style = species] (E) at (1, 0) {$E$};
		\node [style = transition] (b) at (2.5, 0) {$\beta$};
		\node [style = species] (F) at (4,0) {$F$};
		
		\node [style=empty] (X) at (-5.1, 1) {$X$};
		\node [style=none] (Xtr) at (-4.75, 0.75) {};
		\node [style=none] (Xbr) at (-4.75, -0.75) {};
		\node [style=none] (Xtl) at (-5.4, 0.75) {};
             \node [style=none] (Xbl) at (-5.4, -0.75) {};
	
		\node [style=inputdot] (1) at (-5, 0.5) {};
		\node [style=empty] at (-5.2, 0.5) {$1$};
		\node [style=inputdot] (2) at (-5, 0) {};
		\node [style=empty] at (-5.2, 0) {$2$};
		\node [style=inputdot] (3) at (-5, -0.5) {};
		\node [style=empty] at (-5.2, -0.5) {$3$};

		\node [style=empty] (Y) at (-0.1, 1) {$Y$};
		\node [style=none] (Ytr) at (.25, 0.75) {};
		\node [style=none] (Ytl) at (-.4, 0.75) {};
		\node [style=none] (Ybr) at (.25, -0.75) {};
		\node [style=none] (Ybl) at (-.4, -0.75) {};

		\node [style=inputdot] (4) at (0, 0.5) {};
		\node [style=empty] at (0, 0.25) {$4$};
		\node [style=inputdot] (5) at (0, -0.5) {};
		\node [style=empty] at (0, -0.25) {$5$};		
		
		\node [style=empty] (Z) at (5, 1) {$Z$};
		\node [style=none] (Ztr) at (4.75, 0.75) {};
		\node [style=none] (Ztl) at (5.4, 0.75) {};
		\node [style=none] (Zbl) at (5.4, -0.75) {};
		\node [style=none] (Zbr) at (4.75, -0.75) {};

		\node [style=inputdot] (6) at (5, 0) {};
		\node [style=empty] at (5.2, 0) {$6$};	
		
	\end{pgfonlayer}
	\begin{pgfonlayer}{edgelayer}
		\draw [style=inarrow] (A) to (a);
		\draw [style=inarrow] (B) to (a);
		\draw [style=inarrow] (a) to (C);
		\draw [style=inarrow] (a) to (D);
		\draw [style=inarrow] (E) to (b);
		\draw [style=inarrow] (b) to (F);
		\draw [style=inputarrow] (1) to (A);
		\draw [style=inputarrow] (2) to (B);
		\draw [style=inputarrow] (3) to (B);
		\draw [style=inputarrow] (4) to (C);
		\draw [style=inputarrow] (5) to (D);
		\draw [style=inputarrow] (4) to (E);
		\draw [style=inputarrow] (5) to (E);
		\draw [style=inputarrow] (6) to (F);
		\draw [style=simple] (Xtl.center) to (Xtr.center);
		\draw [style=simple] (Xtr.center) to (Xbr.center);
		\draw [style=simple] (Xbr.center) to (Xbl.center);
		\draw [style=simple] (Xbl.center) to (Xtl.center);
		\draw [style=simple] (Ytl.center) to (Ytr.center);
		\draw [style=simple] (Ytr.center) to (Ybr.center);
		\draw [style=simple] (Ybr.center) to (Ybl.center);
		\draw [style=simple] (Ybl.center) to (Ytl.center);
		\draw [style=simple] (Ztl.center) to (Ztr.center);
		\draw [style=simple] (Ztr.center) to (Zbr.center);
		\draw [style=simple] (Zbr.center) to (Zbl.center);
		\draw [style=simple] (Zbl.center) to (Ztl.center);
	\end{pgfonlayer}
\end{tikzpicture}
\]
At this point, if we ignore the sets $X,Y,Z$, we have a new reaction network whose set of species is the disjoint union of those for $R$ and $R'$.  The second step is to identify a species of $R$ with a species of $R'$ whenever both are images of the same point in $Y$.  After this step we can stop drawing everything involving $Y$, and get an open reaction network from $X$ to $Z$:
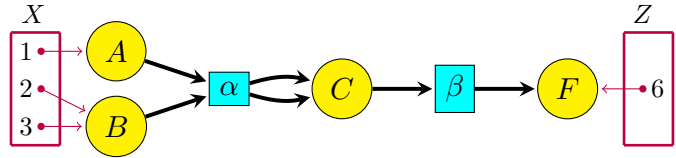
\[
\begin{tikzpicture}
	\begin{pgfonlayer}{nodelayer}
		\node [style=species] (A) at (-4, 0.5) {$A$};
		\node [style=species] (B) at (-4, -0.5) {$B$};;
             \node [style=transition] (a) at (-2.5, 0) {$\alpha$}; 
		\node [style = species] (E) at (-1, 0) {$C$};
		\node [style = transition] (b) at (.5, 0) {$\beta$};
		\node [style = species] (F) at (2,0) {$F$};
		
		\node [style=empty] (X) at (-5.1, 1) {$X$};
		\node [style=none] (Xtr) at (-4.75, 0.75) {};
		\node [style=none] (Xbr) at (-4.75, -0.75) {};
		\node [style=none] (Xtl) at (-5.4, 0.75) {};
             \node [style=none] (Xbl) at (-5.4, -0.75) {};
	
		\node [style=inputdot] (1) at (-5, 0.5) {};
		\node [style=empty] at (-5.2, 0.5) {$1$};
		\node [style=inputdot] (2) at (-5, 0) {};
		\node [style=empty] at (-5.2, 0) {$2$};
		\node [style=inputdot] (3) at (-5, -0.5) {};
		\node [style=empty] at (-5.2, -0.5) {$3$};	
		
		\node [style=empty] (Z) at (3, 1) {$Z$};
		\node [style=none] (Ztr) at (2.75, 0.75) {};
		\node [style=none] (Ztl) at (3.4, 0.75) {};
		\node [style=none] (Zbl) at (3.4, -0.75) {};
		\node [style=none] (Zbr) at (2.75, -0.75) {};

		\node [style=inputdot] (6) at (3, 0) {};
		\node [style=empty] at (3.2, 0) {$6$};	
		
	\end{pgfonlayer}
	\begin{pgfonlayer}{edgelayer}
		\draw [style=inarrow] (A) to (a);
		\draw [style=inarrow] (B) to (a);
	     \draw [style=inarrow, bend right=15, looseness=1.00] (a) to (E);
	     \draw [style=inarrow, bend left =15, looseness=1.00] (a) to (E);		
		\draw [style=inarrow] (E) to (b);
		\draw [style=inarrow] (b) to (F);
		\draw [style=inputarrow] (1) to (A);
		\draw [style=inputarrow] (2) to (B);
		\draw [style=inputarrow] (3) to (B);
		\draw [style=inputarrow] (6) to (F);
		\draw [style=simple] (Xtl.center) to (Xtr.center);
		\draw [style=simple] (Xtr.center) to (Xbr.center);
		\draw [style=simple] (Xbr.center) to (Xbl.center);
		\draw [style=simple] (Xbl.center) to (Xtl.center);
		\draw [style=simple] (Ztl.center) to (Ztr.center);
		\draw [style=simple] (Ztr.center) to (Zbr.center);
		\draw [style=simple] (Zbr.center) to (Zbl.center);
		\draw [style=simple] (Zbl.center) to (Ztl.center);
	\end{pgfonlayer}
\end{tikzpicture}
\]
This new open reaction network is the desired composite $R' R \maps X \to Z$.  Note that we have identified the species $C,D$ and $E$ and arbitrarily named the resulting species $C$.

Formalizing all this, and proving that we get a category whose morphisms are open
reaction networks, is precisely what Fong's theory of decorated cospans  \cite{Fong2015,FongThesis} is designed to accomplish.   A bit of notation is useful: given a finite set $S$, let $F(S)$ be the set of all reaction networks on $S$.  Thus, an open reaction network with rates is a cospan of finite sets with apex $S$ together with an element $R \in F(S)$.  We say that the cospan is \define{decorated} by
$R$.

Now, suppose we have open reaction networks $R \maps X \to Y$ and $R' \maps Y \to Z$ and we wish to define their composite $R' R \maps X \to Z$.   Thus, we have cospans of finite sets
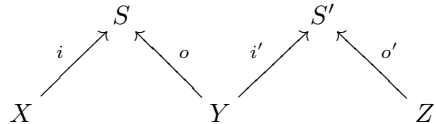
\[ \xymatrix { & S & & S' & \\ X \ar[ur]^{i} & &  Y \ar[ul]_o \ar[ur]^{i'} & & Z \ar[ul]_{o'} } \] 
decorated by elements $R \in F(S)$ and $R' \in F(S')$.   To define their composite, first we 
compose the cospans, and then we compose the decorations.
  
To compose the cospans, first we write them as part of a single diagram:
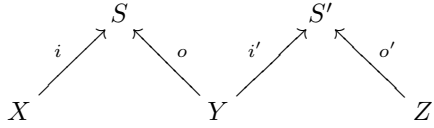
\[ \xymatrix{ & S & & S' & \\ 
X \ar[ur]^{i} & &  Y \ar[ul]_{o} \ar[ur]^{i'} & & Z \ar[ul]_{o'} \\ } \] 
Then we form the pushout of $S$ and $S'$ over $Y$.   This is denoted $S +_Y S'$, and is
formed by first taking the disjoint union of $S$ and $S'$ and then taking the quotient of this by the finest equivalence relation such that $n \sim n'$ if $o(y) = n$ and $i'(y) = n'$ for some $y \in Y$.    This identifies a species of $R$ with one of $R'$ whenever both are images of the same point in $Y$.  The pushout comes with canonical maps $j\maps S \to S+_{Y}S'$ and $j'\maps S' \to S+_{Y}S'$, so we get a diagram
  
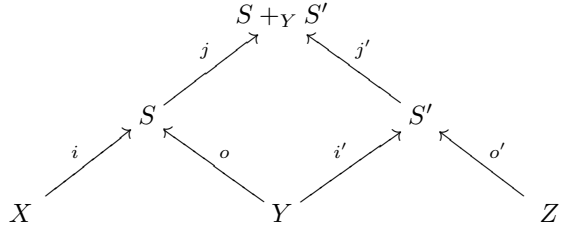
\[
    \xymatrix{
      && S +_Y S' \\
      & S \ar[ur]^{j} && S' \ar[ul]_{j'} \\
      \quad X\quad \ar[ur]^{i} && Y \ar[ul]_{o} \ar[ur]^{i'} &&\quad Z \quad \ar[ul]_{o'}
    }
  \]
and we define the composite cospan to be
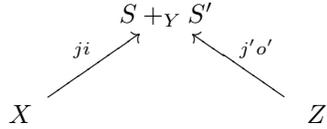
\[ \xymatrix{ & S +_Y S' & \\ \; X \; \ar[ur]^{j i} & & \; \ar[ul]_{j' o'} \; Z. } \]

Next we need to decorate this composite cospan with an element of $F(S +_Y S')$. 
To get this element we start with what we have, namely $R \in F(S)$ and $R' \in F(S')$, and apply two maps.  First, we apply the obvious map
\[   \varphi_{S,S'} \maps F(S) \times F(S') \to F(S + S') \]
that sends a pair of reaction networks, one on $S$ and one on $S'$, to one
on the disjoint union $S + S'$ of these two sets.    Then, we apply a map
\[F([j,j']) \maps F(S+S') \to F(S+_Y S')\]
that glues together the two pieces of this reaction network to get one on 
$S +_Y S'$.  We explain these maps in turn.

First, the `obvious' map $\varphi_{S,S'}$ is defined by
\beq
\label{eq:Phi}
    \varphi_{S,S'} ((S,T,s,t), (S',T',s',t')) = (S + S', T + T', s+s', t+t').  
\eeq
In more detail, for the reaction network on the right hand side:
\begin{itemize}
\item The set of species is the disjoint union $S + S'$.
\item The set of transitions is the disjoint union $T + T'$.
\item The source map $s + s' \maps T + T' \to S + S'$ sends any transition $\tau \in T$ 
to $s(\tau) \in S$ and any transition $\tau' \in T'$ to $s'(\tau') \in S'$.  
\item The target map $t + t' \maps T + T' \to S + S'$ sends any transition $\tau \in T$ 
to $t(\tau) \in S$ and any transition $\tau' \in T'$ to $t'(\tau') \in S'$.  
\end{itemize}

Second, to define the map $F([j,j'])$, we take the canonical map
\[  [j,j'] \maps S+S' \to S+_Y S' \] 
from the disjoint union $S+S'$ onto the pushout, and use the fact that the
$F$ is actually a functor, so it is defined not only on finite sets but also on functions
betwen these:

\begin{lem}
\label{lemma:RFunctor}
There is a functor $F \maps \FinSet \to \Set$ such that:
\begin{itemize}
\item For any finite set $S$, $F(S)$ is the set of all reaction networks on $S$.
\item For any function $f \maps S \to S'$ between finite sets and any reaction network $(S,T,s,t) \in F(S)$, we have
\[    F(f)(S,T,s,t) = (S',T, f_*(s), f_*(t)) \]
where 
\beq  
f_*(s)(\tau)(\sigma') =\sum_{\{\sigma \in S : f(\sigma) = \sigma' \}} s(\tau)(\sigma)  
\label{stilde}
\eeq
and 
\beq  
f_*(t)(\tau)(\tilde{\sigma}) =\sum_{\{\sigma \in S : f(\sigma) = \sigma' \}} t(\tau)(\sigma) .
\label{ttilde}
\eeq
\end{itemize}
\end{lem}
\begin{proof}
To prove that $F$ is a functor we need only check that $F$ preserves composition
and sends identity functions to identity functions. Both are straightforward calculations.
\end{proof}

To show that our procedure for composing open reaction networks gives a category, we use a result of Fong: it suffices to show that $F$ is `lax symmetric monoidal'.  We recall this concept at the beginning of Appendix \ref{sec:deccospan}, and Fong's result in Lemma \ref{lemma:fcospans}.

\begin{lem}
\label{lemma:RLax}
The functor $F$ becomes lax symmetric monoidal from $(\FinSet, +)$ to $(\Set, \times)$ if we equip it with the natural transformation $\varphi_{S,S'} \maps F(S) \times F(S') \to F(S + S')$ defined in Equation \eqref{eq:Phi} and the map $\varphi \maps 1 \to F(\emptyset)$, where $1$ is a chosen one-element set, sending its one element to the unique reaction network with no species and no transitions.
\end{lem}

\begin{proof}
Here $(\FinSet, +)$ denotes $\FinSet$ made into a symmetric monoidal category using  disjoint union as its tensor product, while  $(\Set, \times)$ denotes $\Set$ made into a symmetric monoidal category with the cartesian product as its tensor product.  
To check this lemma we must first show that $\varphi_{S,S'}$ is natural, which is a straightforward calculation.  Commutativity of the diagrams needed for the laxness of $F$ follows from the universal property of the coproduct in $\FinSet$.  \end{proof}

Fong's result \emph{almost} gives a category with finite sets as objects and open reaction networks as morphisms.  However, since the disjoint union of sets is associative only up to isomorphism, so is composition of open reaction networks.   Thus, instead of a category, open reaction networks are really morphisms of a bicategory \cite{Courser}.   To obtain a category, Fong uses certain \emph{equivalence classes} of open reaction networks $R \maps X \to Y$ as morphisms from $X$ to $Y$.  His result then implies:

\begin{thm}
There is a category $\RNet$ where:
\begin{itemize}
\item an object is a finite set,
\item a morphism from $X$ to $Y$ is an equivalence class of open reaction networks 
from $X$ to $Y$, 
\item Given morphisms represented by an open reaction network from $X$ to $Y$ and one from $Y$ to $Z$:
 \[
    (X \stackrel{i}\longrightarrow S \stackrel{o}\longleftarrow Y, R) 
    \quad \textrm{ and } \quad
    (Y \stackrel{i'}\longrightarrow S' \stackrel{o'}\longleftarrow Z, R'), 
  \]
their composite is the equivalence class of this cospan constructed via a pushout:
  
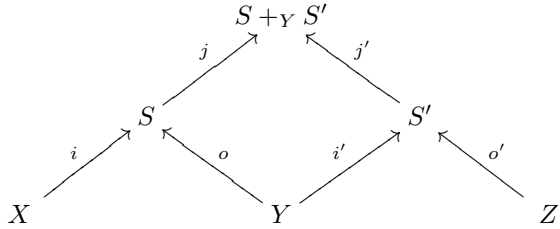
\[
    \xymatrix{
      && S +_Y S' \\
      & S \ar[ur]^{j} && S' \ar[ul]_{j'} \\
      \quad X\quad \ar[ur]^{i} && Y \ar[ul]_{o} \ar[ur]^{i'} &&\quad Z \quad \ar[ul]_{o'}
    }
  \]
together with the reaction network on $S +_Y S'$ obtained by applying the map
\[      
\xymatrix{      F(S) \times F(S') \ar[rr]^-{\varphi_{S,S'}} && 
                     F(S + S') \ar[rr]^-{F([j,j'])} && F(S +_Y S') } \]
to the pair $(R,R') \in F(S) \times F(S')$.  
\end{itemize}
\end{thm}

\begin{proof}
This follows from Lemma \ref{lemma:RLax} together with Lemma 
\ref{lemma:fcospans}, where we explain the equivalence relation in detail.  
\end{proof}

In fact, Fong's machinery proves more.  We can take the `tensor product' of  two open reaction networks by setting them side by side.  They then act in parallel with no interaction between them.  This makes $\RNet$ into symmetric monoidal category. 
In fact it is one of a very nice sort, called a `hypergraph category'.  For more on this concept, see Fong's thesis \cite{FongThesis}.

\begin{thm}
The category $\RNet$ is a symmetric monoidal category where the tensor product of objects $X$ and $Y$ is their disjoint union $X + Y$, while the tensor product of the morphisms
\[
    (X \stackrel{i}{\longrightarrow} S \stackrel{o}{\longleftarrow} Y, R) 
    \quad \textrm{ and } \quad
    (X' \stackrel{i'}{\longrightarrow} S' \stackrel{o'}{\longleftarrow} Y', R') 
  \]
is defined to be
\[  ( X + X' \stackrel{i+i'}{\longrightarrow} S + S'  \stackrel{o + o'}{\longleftarrow} Y + Y', \;
\varphi_{S,S'}(R,R') ) .\]
In fact $\RNet$ is a hypergraph category, and thus a dagger-compact category.
\end{thm}

\begin{proof}
This follows from Theorem 3.4 of Fong's paper on decorated cospans \cite{Fong2015}.
\end{proof}

\section{Open reaction networks with rates}
\label{sec:RxNet}

To get a dynamical system from an open reaction network we need to equip its transitions with nonnegative real numbers called `rate constants':

\begin{defn}
A \define{reaction network with rates} $R = (S,T,s,t,r)$ consists of a reaction network
$(S,T,s,t)$ together with a function $r \maps T \to (0,\infty)$ specifying a \define{rate constant} for each transition.   We call any reaction network with rates having $S$ as its set of species a reaction network with rates \define{on} $S$.
\end{defn}

\noindent
Just as reaction networks are equivalent to Petri nets, reaction networks with rates
are equivalent to Petri nets where each transition is equipped with a rate constant.  
These are usually called `stochastic Petri nets', because they can be used to define 
stochastic processes \cite{BaezBiamonte, GossPeccoud, Haas}.   

The results of the last section easily generalize to reaction networks with rates: 

\begin{defn}
Given finite sets $X$ and $Y$, an \define{open reaction network with rates from} $X$ \define{to} $Y$ is a cospan of finite sets
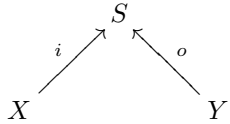
\[ \xymatrix{  & S  &  \\ X \ar[ur]^{i} & & Y \ar[ul]_{o} } \] 
together with a reaction network with rates on $R$ on $S$. We often abbreviate all this data as $R \maps X \to Y$. 
\end{defn}

\begin{lem}
\label{lemma:RxFunctor}
There is a functor $F \maps \FinSet \to \Set$ such that:
\begin{itemize}
\item For any finite set $S$, $F(S)$ is the set of all reaction networks with rates on $S$.
\item For any function $f \maps S \to S'$ between finite sets and any reaction network with rates $(S,T,s,t,r) \in F(S)$, we have
\[    F(f)(S,T,s,t,r) = (S',T, f_*(s), f_*(t), r) \]
where $f_*(s)$ and $f_*(t)$ are defined as in Equations \eqref{stilde} and \eqref{ttilde}.
\end{itemize}
\end{lem}

\begin{proof}
This is a slight variation on Lemma \ref{lemma:RFunctor}.
\end{proof}

\begin{lem}
\label{lemma:RxLax}
The functor $F$ can be made lax symmetric monoidal from $(\FinSet, +)$ to $(\Set, \times)$.  To do this we equip it with the map $\varphi \maps 1 \to F(\emptyset)$ sending the one element of $1$ to the unique reaction network with rates having no species and no transitions, together with natural transformation $\varphi_{S,S'} \maps F(S) \times F(S') \to F(S + S')$ 
such that
\[
    \varphi_{S,S'} ((S,T,s,t,r), (S',T',s',t',r')) = (S + S', T + T', s+s', t+t',[r,r']).  
\]
where:
\begin{itemize}
\item The map $s + s' \maps T + T' \to S + S'$ sends any transition $\tau \in T$ 
to $s(\tau) \in S$ and any transition $\tau' \in T'$ to $s'(\tau') \in S'$.  
\item The map $t + t' \maps T + T' \to S + S'$ sends any transition $\tau \in T$ 
to $t(\tau) \in S$ and any transition $\tau' \in T'$ to $t'(\tau') \in S'$.  
\item The map $[r,r'] \maps T + T' \to [0,\infty)$ sends any transition $\tau \in T$ to $r(\tau)$ and any transition $\tau' \in T'$ to $r'(\tau')$.  
\end{itemize}
\end{lem}

\begin{proof}
This is a slight variation on Lemma \ref{lemma:RLax}. 
\end{proof}

We now come to the star of the show, the category $\RxNet$:

\begin{thm}
\label{thm:RxNet}
There is a category $\RxNet$ where:
\begin{itemize}
\item an object is a finite set,
\item a morphism from $X$ to $Y$ is an equivalence class of open reaction networks 
with rates from $X$ to $Y$, 
\item Given morphisms represented by an open reaction network with rates from $X$ to $Y$ and one from $Y$ to $Z$:
 \[
    (X \stackrel{i}\longrightarrow S \stackrel{o}\longleftarrow Y, R) 
    \quad \textrm{ and } \quad
    (Y \stackrel{i'}\longrightarrow S' \stackrel{o'}\longleftarrow Z, R'), 
  \]
their composite consists of the equivalence class of this cospan:
  \[
    \xymatrix{
      & S +_Y S' \\
      \quad X\quad \ar[ur]^{ji} && \quad Z \quad \ar[ul]_{j'o'}
    }
  \]
together with the reaction network with rates on $S +_Y S'$ 
obtained by applying the map
\[      
\xymatrix{      F(S) \times F(S') \ar[rr]^-{\varphi_{S,S'}} && 
                     F(S + S') \ar[rr]^-{F([j,j'])} && F(S +_Y S') } \]
to the pair $(R,R') \in F(S) \times F(S')$.  
\end{itemize}
The category $\RxNet$ is a symmetric monoidal category where the tensor product of objects $X$ and $Y$ is their disjoint union $X + Y$, while the tensor product of the morphisms
\[
    (X \stackrel{i}{\longrightarrow} S \stackrel{o}{\longleftarrow} Y, R) 
    \quad \textrm{ and } \quad
    (X' \stackrel{i'}{\longrightarrow} S' \stackrel{o'}{\longleftarrow} Y', R') 
  \]
is defined to be
\[  ( X + X' \stackrel{i+i'}{\longrightarrow} S + S'  \stackrel{o + o'}{\longleftarrow} Y + Y', \; \varphi_{S,S'}(R,R') ) .\]
In fact $\RxNet$ is a hypergraph category.
\end{thm}

\begin{proof}
This follows from Lemmas \ref{lemma:RxLax} and \ref{lemma:fcospans}, 
where we explain the equivalence relation in detail.  
\end{proof}

\section{The open rate equation}
\label{sec:openrate}

In chemistry, a reaction network with rates is frequently used as a tool to specify a dynamical system.   A dynamical system is often defined as a smooth manifold $M$ whose points are `states', together with a smooth vector field on $M$ saying how these states evolve in time.   In chemistry we take $M = [0,\infty)^S$ where $S$ is the set of species: a point $c \in [0,\infty)^S$ describes the \define{concentration} $c_\sigma$ of each species $\sigma \in S$.   The corresponding dynamical system is a first-order differential equation called the `rate equation':
\[   \frac{dc(t)}{dt} = v(c(t)) \]
where now $c \maps \R \to [0,\infty)^S$ describes the concentrations as a function of 
time, and $v$ is a vector field on $[0,\infty)^S$.   Of course, $[0,\infty)^S$ is not a smooth manifold.  On the other hand, the vector field $v$ is better than smooth: its components are polynomials, so it is \define{algebraic}.   For mathematical purposes this lets us treat $v$ as an vector field on all of $\R^S$, even though negative concentrations are unphysical.  

In more detail, suppose $R = (S,T,s,t,r)$ is a reaction network with rates. 
Then the rate equation is determined by a rule called the \define{law of mass action}. 
This says that each transition $\tau \in T$ contributes 
to $dc(t)/dt$ by the product of:
\begin{itemize}
\item the rate constant $r(\tau)$, 
\item the concentration of each species $\sigma$ raised to the power given by the number of times $\sigma$ appears as an input to $\tau$, namely $s(\tau)(\sigma)$, and
\item the vector $t(\tau) - s(\tau) \in \R^S$ whose $\sigma$th component is the change
in the number of items of the species $\sigma \in S$ caused by the transition $\tau$,
\end{itemize}
The second factor is a product over all species, and it deserves an abbreviated notation: 
given $c \in \R^S$ we write
\begin{equation}
\label{eq:power}
 c^{s(\tau)} = \prod_{\sigma \in S} {c_\sigma}^{s(\tau)(\sigma)} .
\end{equation}
Thus:

\begin{defn} 
We say the time-dependent concentrations $c \maps \R \to \R^S$ obey the
\define{rate equation} for the reaction network with rates $R = (S,T,s,t,r)$ if
\[
\frac{dc(t)}{d t} = 
\sum_{\tau \in T} r(\tau) (t(\tau) - s(\tau)) c(t)^{s(\tau)} .
\]
\end{defn}

\noindent
For short, we can write the rate equation as
\[    \frac{dc(t)}{d t} = v^R(c(t)) \]
where $v^R$ is the vector field on $\R^S$ given by
\begin{equation}
\label{eq:v^R}
v^R(c) = \sum_{\tau \in T} r(\tau) \, ( t(\tau) - s(\tau) ) c^{s(\tau)} 
\end{equation}
at any point $c \in \R^S$.  We call the components of this vector field \define{reaction velocities}, since in the rate equation they describe rates of change of concentrations.   

Given an \emph{open} reaction network with rates, we can go further: we can obtain an \emph{open} dynamical system.   We give a specialized definition of this concept suited to the case at hand:

\begin{defn}
Given finite sets $X$ and $Y$, an \define{open dynamical system from} $X$ \define{to} $Y$ is a cospan of finite sets
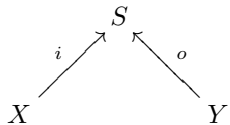
\[ \xymatrix{  & S  &  \\ X \ar[ur]^{i} & & Y \ar[ul]_{o} } \] 
together with an algebraic vector field $v$ on $\R^S$. 
\end{defn}

\noindent
The point is that given an open dynamical system of this sort, we can write down a generalization of the rate equation that takes into account `inflows' and `outflows' as well as the intrinsic dynamics given by the vector field $v$.   

To make this precise, let the
\define{inflows} $I \maps \R \to \R^X $ and \define{outflows} $O \maps \R \to \R^Y$ be arbitrary smooth functions of time.   We write the inflow at the point $x \in X$ as $I_x(t)$ or simply $I_x$, and similarly for the outflows.   Given an open dynamical system and a choice of inflows and outflows, we define the pushforward $i_*(I) \maps \R \to \R^S$ by
\[ i_*(I)_\sigma = \sum_{ \{ x : i(x) = \sigma \} } I_x \]
and define $o_*(O) \maps \R \to \R^S$ by
\[   o_*(O)_\sigma = \sum_{ \{ y : o(y) = \sigma \} } O_y \]
With this notation, the \define{open rate equation} is
\[ \frac{dc(t)}{dt} = v(c(t)) + i_*(I(t)) - o_*(O(t)). \]
The pushforwards here say that for any species $\sigma \in S$, the time derivative of 
the concentration $c_\sigma(t)$ takes into account the sum of all inflows at $x \in X$ such that $i(x) = \sigma$, minus the sum of outflows at $y \in Y$ such that $o(y) = \sigma$.  

To make these ideas more concrete, let us see in an example how to go from an open reaction network with rates to an open dynamical system and then its open rate equation.  Let $R$ be the following reaction network:
\[   \xymatrix{ A+B \ar[r]^{\tau} & C + D. }  \] 
The set of species is $S = \{A,B,C,D\}$ and the set of transitions is just $T = \{\tau\}$.   We can make $R$ into a reaction network with rates by saying the rate constant of $\tau$ is some positive number $r$.   This gives a vector field 
\[      v^R(A,B,C,D) = (-r A B, -r A B, r A B, r A B) \]
where we abuse notation in a commonly practiced way and use $(A,B,C,D)$ as the  coordinates for a point in $\R^S$: that is, the concentrations of the four species
with the same names.  The resulting rate equation is
\[ 
\begin{array}{rcl} 
\displaystyle{\frac{dA(t)}{dt}} &=& - r A(t) B(t) \\ \\
\displaystyle{\frac{dB(t)}{dt}} &=& - r A(t) B(t) \\ \\
\displaystyle{\frac{dC(t)}{dt}}&=& r A(t) B(t) \\ \\
\displaystyle{\frac{dD(t)}{dt}} &=& r A(t) B(t) .
\end{array}
\]
Next, we can make $R$ into an open reaction network $R \maps X \to Y$ as follows:
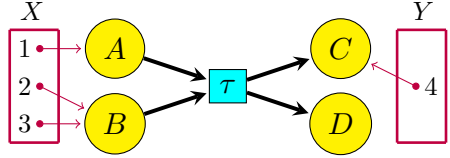
\[
\begin{tikzpicture}
	\begin{pgfonlayer}{nodelayer}
		\node [style=species] (A) at (-4, 0.5) {$A$};
		\node [style=species] (B) at (-4, -0.5) {$B$};
		\node [style=species] (C) at (-1, 0.5) {$C$};
		\node [style=species] (D) at (-1, -0.5) {$D$};
             \node [style=transition] (a) at (-2.5, 0) {$\tau$}; 
		
		\node [style=empty] (X) at (-5.1, 1) {$X$};
		\node [style=none] (Xtr) at (-4.75, 0.75) {};
		\node [style=none] (Xbr) at (-4.75, -0.75) {};
		\node [style=none] (Xtl) at (-5.4, 0.75) {};
             \node [style=none] (Xbl) at (-5.4, -0.75) {};
	
		\node [style=inputdot] (1) at (-5, 0.5) {};
		\node [style=empty] at (-5.2, 0.5) {$1$};
		\node [style=inputdot] (2) at (-5, 0) {};
		\node [style=empty] at (-5.2, 0) {$2$};
		\node [style=inputdot] (3) at (-5, -0.5) {};
		\node [style=empty] at (-5.2, -0.5) {$3$};

		\node [style=empty] (Y) at (0.1, 1) {$Y$};
		\node [style=none] (Ytr) at (.4, 0.75) {};
		\node [style=none] (Ytl) at (-.25, 0.75) {};
		\node [style=none] (Ybr) at (.4, -0.75) {};
		\node [style=none] (Ybl) at (-.25, -0.75) {};

		\node [style=inputdot] (4) at (0, 0) {};
		\node [style=empty] at (0.2, 0) {$4$};
		
		
	\end{pgfonlayer}
	\begin{pgfonlayer}{edgelayer}
		\draw [style=inarrow] (A) to (a);
		\draw [style=inarrow] (B) to (a);
		\draw [style=inarrow] (a) to (C);
		\draw [style=inarrow] (a) to (D);
		\draw [style=inputarrow] (1) to (A);
		\draw [style=inputarrow] (2) to (B);
		\draw [style=inputarrow] (3) to (B);
		\draw [style=inputarrow] (4) to (C);
		\draw [style=simple] (Xtl.center) to (Xtr.center);
		\draw [style=simple] (Xtr.center) to (Xbr.center);
		\draw [style=simple] (Xbr.center) to (Xbl.center);
		\draw [style=simple] (Xbl.center) to (Xtl.center);
		\draw [style=simple] (Ytl.center) to (Ytr.center);
		\draw [style=simple] (Ytr.center) to (Ybr.center);
		\draw [style=simple] (Ybr.center) to (Ybl.center);
		\draw [style=simple] (Ybl.center) to (Ytl.center);
	\end{pgfonlayer}
\end{tikzpicture}
\]
Here $X = \{1,2,3\}$ and $Y = \{4,5\}$, while the functions $i \maps X \to S$ and $o \maps Y \to S$ are given by
\[      i(1) = A, \; i(2) = i(3) = B, \; o(4) = C. \]
The corresponding open dynamical system is the cospan
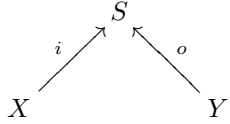
\[ \xymatrix{  & S  &  \\ X \ar[ur]^{i} & & Y \ar[ul]_{o} } \]
decorated by the vector field $v^R$ on $\R^S$.   Finally, the corresponding open rate equation is
\[ 
\begin{array}{rcl} 
\displaystyle{\frac{dA(t)}{dt}} &=& - r A(t) B(t)  + I_1(t)\\ \\
\displaystyle{\frac{dB(t)}{dt}} &=& - r A(t) B(t) + I_2(t) + I_3(t) \\ \\
\displaystyle{\frac{dC(t)}{dt}} &=& r A(t) B(t) - O_4(t) \\ \\
\displaystyle{\frac{dD(t)}{dt}} &=& r A(t) B(t).
\end{array}
\]
Note that $dB/dt$ involves the sum of two inflow terms $I_1$ and $I_2$ 
since $i(2) = i(3) = B$, while $dD/dt$ involves neither inflow nor outflow terms 
since $D$ is in the range of neither $i$ nor $o$.  

\section{The category of open dynamical systems}
\label{sec:opendynam}

There is a category $\Dynam$ where the morphisms are open dynamical systems ----or more precisely, certain equivalence classes of these.  We compose two open dynamical systems by connecting the outputs of the first to the inputs of the second.  

To construct the category $\Dynam$, we again use the machinery of decorated cospans.  
For this we need a lax monoidal functor $D \maps \FinSet \to \Set$ sending any finite set $S$ to the set of algebraic vector fields on $\R^S$.

\begin{lem}
\label{lemma:Dfunctor}
There is a functor $D \maps \FinSet \to \Set$ such that:
\begin{itemize}
\item $D$ maps any finite set $S$ to 
\[ D(S) = \{ v \maps \R^S \to \R^S : \; v \textrm{ is algebraic}  \}. \] 
\item $D$ maps any function $f \maps S \to S'$ between finite sets to the function $D(f) \maps D(S) \to D(S')$ given as follows:
\[ D(f)(v) = f_* \circ v \circ f^* \]
where the pullback $ f^* \maps \R^{S'} \to \R^S $ is given by
\[ f^*(c)(\sigma) = c(f(\sigma)) \] 
while the pushforward $ f_* \maps \R^{S} \to \R^{S'} $ is given by
\[ f_*(c)(\sigma') = \sum_{ \{ \sigma \in S : f(\sigma) = \sigma' \} } c(\sigma). \]
\end{itemize}
\end{lem}

\begin{proof}
The functoriality of $D$ follows from the fact that pushforward is a covariant functor and pullback is a contravariant functor:
\[   D(f)D(g)(v) = f_* \circ g_* \circ v \circ g^* \circ f^* = (f\circ g)_* \circ v \circ (f\circ g)^* = D(fg)(v).  \qedhere \]
\end{proof}

\begin{lem}
\label{lemma:DLax}
The functor $D$ becomes lax symmetric monoidal from $(\FinSet, +)$ to $(\Set, \times)$ if we equip it with the natural transformation 
\[ \delta_{S,S'} \maps D(S) \times D(S') \to D(S + S') \]
given by
\[  \delta_{S,S'}(v,v') = i_* \circ v \circ i^* + i'_* \circ v' \circ {i'}^* \]
together with the unique map $\delta \maps 1 \to D(\emptyset)$.
Here $i \maps S \to S+S'$ and $i' \maps S' \to S+S'$ are the inclusions of $S$ and $S'$ into their disjoint union, and we add vector fields in the usual way.
\end{lem}

\begin{proof}
By straightforward calculations one can verify all the conditions in the definition of lax symmetric monoidal functor, Def.\ \ref{defn.lmf}.   \end{proof}

\begin{thm}
\label{thm:dynam}
There is a category $\Dynam$ where:
\begin{itemize}
\item an object is a finite set,
\item a morphism from $X$ to $Y$ is an equivalence class of open dynamical
systems from $X$ to $Y$, 
\item Given an open dynamical system from $X$ to $Y$ and one from $Y$ to $Z$:
 \[
    (X \stackrel{i}\longrightarrow S \stackrel{o}\longleftarrow Y, v) 
    \quad \textrm{ and } \quad
    (Y \stackrel{i'}\longrightarrow S' \stackrel{o'}\longleftarrow Z, v'), 
  \]
their composite consists of the equivalence class of this cospan:
  \[
    \xymatrix{
      & S +_Y S' \\
      \quad X\quad \ar[ur]^{ji} && \quad Z \quad \ar[ul]_{j'o'}
    }
  \]
together with the algebraic vector field on $\R^{S+_Y S'}$ obtained by applying the map
\[      
\xymatrix{      D(S) \times D(S') \ar[rr]^-{\delta_{S,S'}} && 
                     D(S + S') \ar[rr]^-{D([j,j'])} && D(S +_Y S') } \]
to the pair $(v,v') \in D(S) \times D(S')$.  
\end{itemize}
The category $\Dynam$ is a symmetric monoidal category where the tensor product of objects $X$ and $Y$ is their disjoint union $X + Y$, while the tensor product of the morphisms
\[
    (X \stackrel{i}{\longrightarrow} S \stackrel{o}{\longleftarrow} Y, v) 
    \quad \textrm{ and } \quad
    (X' \stackrel{i'}{\longrightarrow} S' \stackrel{o'}{\longleftarrow} Y', v') 
  \]
is defined to be
\[  ( X + X' \stackrel{i+i'}{\longrightarrow} S + S'  \stackrel{o + o'}{\longleftarrow} Y + Y', \; \delta_{S,S'}(v,v') ) .\]
In fact $\Dynam$ is a hypergraph category.
\end{thm}

\begin{proof}
This follows from Lemmas \ref{lemma:DLax} and \ref{lemma:fcospans}.
\end{proof}

\section{The gray-boxing functor}
\label{sec:gray}

Now we are ready to describe the `gray-boxing' functor $\graysquare \maps \RxNet \to \Dynam$.  This sends any open reaction network to the open dynamical system that it determines.  The functoriality of this process says that we can first compose networks and then find the open dynamical system of the resulting larger network, or first find the open dynamical system for each network and then compose these systems: either way, the result is the same.

To construct the gray-boxing functor we again turn to Fong's theory of decorated
cospans.  Just as this theory gives decorated cospan categories from lax symmetric 
monoidal functors, it gives functors between such categories from monoidal natural transformations.  For details, see Theorem \ref{thm:decoratedfunctors} in Appendix \ref{sec:deccospan}.

\begin{thm}
There is a symmetric monoidal functor $\graysquare \maps \RxNet \to \Dynam$ that is
the identity on objects and sends each morphism represented by an open reaction network $( X \stackrel{i}{\to} S \stackrel{o}{\leftarrow} Y , R )$ to the morphism represented by the open dynamical system $( X \stackrel{i}{\to} S \stackrel{o}{\leftarrow} Y , v^R )$, where $v^R$ is defined by Equation \ref{eq:v^R}.  Moreover, $\graysquare$ is a hypergraph functor.
\end{thm}

\begin{proof}
Recall that the functors $F,D \maps (\FinSet,+) \to (\Set,\times)$ assign to a set $S$ the set of all possible reaction networks on $S$ and the set of all algebraic vector fields on $\R^S$, respectively.  By Theorem \ref{thm:decoratedfunctors}, we may obtain a hypergraph functor $\graysquare \maps \RxNet \to \Dynam$ from a monoidal natural transformation $\theta_S \maps F(S) \to D(S)$.    For any $R \in F(S)$, let us define $\theta_S(R) \in D(S)$ by
\[  \theta_S(R) = v^R \]
where the vector field $v^R$ is given by Equation \ref{eq:v^R}.

To check the naturality of $\theta$, we must prove that the following square commutes:
\[  \xymatrix{   F(S) \ar_{\theta_S}[d] \ar^{F(f)}[r] & F(S') \ar^{\theta_{S'} }[d] \\
D(S) \ar_{D(f)}[r] & D(S') } \]
where $F(f)$ was defined in Lemma \ref{lemma:RxFunctor} and $D(f)$ was defined in Lemma \ref{lemma:Dfunctor}.   So, consider any element $R \in F(S)$: that is, any reaction network with 
rates 
\[    R = (S,T,s,t,r) .\]
Let $R' = F(f)(R)$.   Thus,
\[   R' = (S',T,f_*(s),f_*(t),r).\]
We need to check that $D(f)(v^R) = v^{R'}$, or in other words,
\[    f_* \circ v^R \circ f^* = v^{R'}  .\]

To do this, recall from Equation \ref{eq:v^R} that for any concentrations $c' \in \R^{S'}$ we have
\[    v^{R'}(c') = \sum_{\tau \in T} r(\tau) (f_*(t)(\tau) - f_*(s)(\tau)) \;{c'}^{f_*(s)(\tau)} .\]
Using Equation (\ref{eq:power}) and the definition of pushforward and pullback, we obtain
\[  
\begin{array}{ccl}
{c'}^{f_*(s)(\tau)} &=& \displaystyle{\prod_{\sigma' \in S'} {c'_{\sigma'}}^{f_*(s)(\tau)(\sigma')} }
\\ \\ 
&=& \displaystyle{ \prod_{\sigma' \in S'} {c'_{\sigma'}}^{\sum_{\{\sigma : \; f(\sigma) = \sigma'\}} s(\tau)(\sigma)} } \\ \\
&=& \displaystyle{ \prod_{\sigma' \in S'}  \prod_{\{\sigma : \; f(\sigma) = \sigma'\}} 
{c'_{\sigma'}}^{s(\tau)(\sigma)} } \\ \\
&=& \displaystyle{  \prod_{\sigma \in S} {c'_{f(\sigma)}}^{s(\tau)(\sigma)} }  \\ \\
&=& \displaystyle{  \prod_{\sigma \in S} {f^*(c')_\sigma}^{s(\tau)(\sigma)} }  \\ \\
&=& f^*(c')^{s(\tau)} .
\end{array} 
\]
Thus, 
\[
\begin{array}{ccl}
   v^{R'}(c') &=&
\displaystyle{ \sum_{\tau \in T} r(\tau) (f_*(t)(\tau) - f_*(s)(\tau))  \; f^*(c')^{s(\tau)} } \\ \\
&=& f_*(v^R (f^*(c'))) .
\end{array}
\]
so $v^{R'} = f_* \circ v^R \circ f^*$ as desired.

We must also check that $\theta$ is monoidal.   By Definition \ref{def:monnattran}, this means showing
that
   \[
    \xymatrix{
      F(S) \times F(S') \ar[r]^-{\varphi_{S,S'}} \ar[d]_{\theta_S \times
      \theta_{S'}} & F(S + S') \ar[d]^{\theta_{S + S'}} \\
      D(S) \times D(S') \ar[r]^-{\delta_{S,S'}} & D(S+S')
    }
  \]
commutes for all $S,S'\in \FinSet$, where $\varphi$ was defined in Lemma \ref{lemma:RxLax} and
$\delta$ was defined in Lemma \ref{lemma:DLax}.   This is straightforward.
\end{proof}

The idea behind this theorem is best explained with an example.  Consider two composable open reaction networks with rates.   The first, $R \maps X \to Y$, is this:
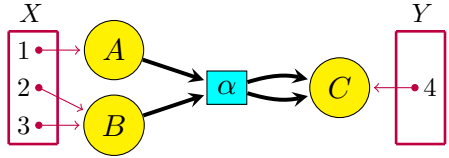
\[
\begin{tikzpicture}
	\begin{pgfonlayer}{nodelayer}
		\node [style=species] (A) at (-4, 0.5) {$A$};
		\node [style=species] (B) at (-4, -0.5) {$B$};
		\node [style=species] (C) at (-1, 0) {$C$};
             \node [style=transition] (a) at (-2.5, 0) {$\alpha$}; 
		
		\node [style=empty] (X) at (-5.1, 1) {$X$};
		\node [style=none] (Xtr) at (-4.75, 0.75) {};
		\node [style=none] (Xbr) at (-4.75, -0.75) {};
		\node [style=none] (Xtl) at (-5.4, 0.75) {};
             \node [style=none] (Xbl) at (-5.4, -0.75) {};
	
		\node [style=inputdot] (1) at (-5, 0.5) {};
		\node [style=empty] at (-5.2, 0.5) {$1$};
		\node [style=inputdot] (2) at (-5, 0) {};
		\node [style=empty] at (-5.2, 0) {$2$};
		\node [style=inputdot] (3) at (-5, -0.5) {};
		\node [style=empty] at (-5.2, -0.5) {$3$};

		\node [style=empty] (Y) at (0.1, 1) {$Y$};
		\node [style=none] (Ytr) at (.4, 0.75) {};
		\node [style=none] (Ytl) at (-.25, 0.75) {};
		\node [style=none] (Ybr) at (.4, -0.75) {};
		\node [style=none] (Ybl) at (-.25, -0.75) {};

		\node [style=inputdot] (4) at (0, 0) {};
		\node [style=empty] at (0.2, 0) {$4$};
		
	\end{pgfonlayer}
	\begin{pgfonlayer}{edgelayer}
		\draw [style=inarrow] (A) to (a);
		\draw [style=inarrow] (B) to (a);
		\draw [style=inarrow, bend left =15] (a) to (C);
		\draw [style=inarrow, bend right =15] (a) to (C);
		\draw [style=inputarrow] (1) to (A);
		\draw [style=inputarrow] (2) to (B);
		\draw [style=inputarrow] (3) to (B);
		\draw [style=inputarrow] (4) to (C);
		\draw [style=simple] (Xtl.center) to (Xtr.center);
		\draw [style=simple] (Xtr.center) to (Xbr.center);
		\draw [style=simple] (Xbr.center) to (Xbl.center);
		\draw [style=simple] (Xbl.center) to (Xtl.center);
		\draw [style=simple] (Ytl.center) to (Ytr.center);
		\draw [style=simple] (Ytr.center) to (Ybr.center);
		\draw [style=simple] (Ybr.center) to (Ybl.center);
		\draw [style=simple] (Ybl.center) to (Ytl.center);
	\end{pgfonlayer}
\end{tikzpicture}
\]
It has species $S = \{A,B,C\}$ and transitions $T = \{\alpha\}$.  
The vector field describing its dynamics is 
\begin{equation}
\label{eq_v_1}
 v^R(A,B,C) = ( -r(\alpha) AB, -r(\alpha) AB , 2r(\alpha) AB). 
\end{equation}
The corresponding open rate equation is 
\begin{equation}
\label{eq:open_rate_1}
\begin{array}{rcl} 
\displaystyle{\frac{dA(t)}{dt}} &=& - r(\alpha) A(t) B(t)  + I_1(t)\\ \\
\displaystyle{\frac{dB(t)}{dt}} &=& - r(\alpha) A(t) B(t) + I_2(t) + I_3(t) \\ \\
\displaystyle{\frac{dC(t)}{dt}} &=& 2r(\alpha) A(t) B(t) - O_4(t).
\end{array}
\end{equation}

The second open reaction network with rates, $R' \maps Y \to Z$, is this:
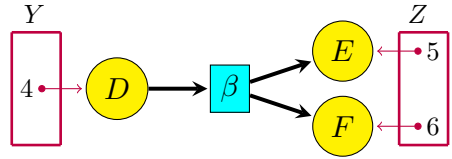
\[
\begin{tikzpicture}
	\begin{pgfonlayer}{nodelayer}
		\node [style = species] (D) at (1, 0) {$D$};
		\node [style = transition] (b) at (2.5, 0) {$\beta$};
		\node [style = species] (E) at (4,0.5) {$E$};
		\node [style = species] (F) at (4,-0.5) {$F$};

		\node [style=empty] (Y) at (-0.1, 1) {$Y$};
		\node [style=none] (Ytr) at (.25, 0.75) {};
		\node [style=none] (Ytl) at (-.4, 0.75) {};
		\node [style=none] (Ybr) at (.25, -0.75) {};
		\node [style=none] (Ybl) at (-.4, -0.75) {};

		\node [style=inputdot] (4) at (0, 0) {};
		\node [style=empty] at (-0.2, 0) {$4$};
		
		\node [style=empty] (Z) at (5, 1) {$Z$};
		\node [style=none] (Ztr) at (4.75, 0.75) {};
		\node [style=none] (Ztl) at (5.4, 0.75) {};
		\node [style=none] (Zbl) at (5.4, -0.75) {};
		\node [style=none] (Zbr) at (4.75, -0.75) {};

		\node [style=inputdot] (5) at (5, 0.5) {};
		\node [style=empty] at (5.2, 0.5) {$5$};	
		\node [style=inputdot] (6) at (5, -0.5) {};
		\node [style=empty] at (5.2, -0.5) {$6$};	

	\end{pgfonlayer}
	\begin{pgfonlayer}{edgelayer}
		\draw [style=inarrow] (D) to (b);
		\draw [style=inarrow] (b) to (E);
		\draw [style=inarrow] (b) to (F);
		\draw [style=inputarrow] (4) to (D);
		\draw [style=inputarrow] (5) to (E);
		\draw [style=inputarrow] (6) to (F);
		\draw [style=simple] (Ytl.center) to (Ytr.center);
		\draw [style=simple] (Ytr.center) to (Ybr.center);
		\draw [style=simple] (Ybr.center) to (Ybl.center);
		\draw [style=simple] (Ybl.center) to (Ytl.center);
		\draw [style=simple] (Ztl.center) to (Ztr.center);
		\draw [style=simple] (Ztr.center) to (Zbr.center);
		\draw [style=simple] (Zbr.center) to (Zbl.center);
		\draw [style=simple] (Zbl.center) to (Ztl.center);
	\end{pgfonlayer}
\end{tikzpicture}
\]
It has species $S'=\{D,E,F\}$ and transitions $T' = \{\beta\}$.    The vector field describing its dynamics is
\begin{equation}
\label{eq:v_2}
 v^{R'}(D,E,F) = ( -r(\beta) D , r(\beta) D , r(\beta) D). 
\end{equation}
The corresponding open rate equation is
\begin{equation}
\label{eq:open_rate_2}
\begin{array}{rcl} 
\displaystyle{\frac{dD(t)}{dt}} &=& - r(\beta) D(t)  + I_4(t)\\ \\
\displaystyle{\frac{dE(t)}{dt}} &=& r(\beta) D(t) - O_5(t) \\ \\
\displaystyle{\frac{dF(t)}{dt}} &=& r(\beta) D(t) - O_6(t).
\end{array}
\end{equation}

Composing $R$ and $R'$ gives $R' R \maps X \to Z$, which looks like this:
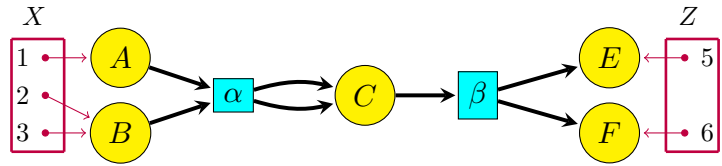
\[
\begin{tikzpicture}
	\begin{pgfonlayer}{nodelayer}
		\node [style=inputdot] (0) at (-4.25, 0) {};
		\node [style=empty] at (-4.55, 0) {$2$};
		\node [style=species] (1) at (-3.25, 0.5) {$A$};
		\node [style=none] (2) at (-4, 0.75) {};
		\node [style=none] (3) at (4, -0.75) {};
		\node [style=transition] (4) at (-1.75, -0) {$\alpha$};
		\node [style=none] (5) at (-4.7, 0.75) {};
		\node [style=none] (6) at (4, 0.75) {};
		\node [style=transition] (7) at (1.5, -0) {$\beta$};
		\node [style=inputdot] (8) at (4.25, 0.5) {};
		\node [style=empty] at (4.55, 0.5) {$5$};
		\node [style=none] (9) at (-4.7, -0.75) {};
		\node [style=species] (10) at (0, -0) {$C$};
		\node [style=none] (11) at (4.7, -0.75) {};
		\node [style=inputdot] (12) at (-4.25, 0.5) {};
		\node [style=empty] at (-4.55, 0.5) {$1$};
		\node [style=none] (13) at (4.7, 0.75) {};
		\node [style=empty] (14) at (-4.4, 1.05) {$X$};
		\node [style=none] (15) at (-4, -0.75) {};
		\node [style=empty] (16) at (4.3, 1.05) {$Z$};
		\node [style=species] (17) at (3.25, 0.5) {$E$};
		\node [style=species] (18) at (-3.25, -0.5) {$B$};
		\node [style=inputdot] (19) at (-4.25, -0.5) {};
		\node [style=empty] at (-4.55, -0.5) {$3$};
		\node[ style=species] (20) at (3.25, -0.5) {$F$};
		\node[ style=inputdot] (21) at (4.25, -0.5) {};
		\node [style=empty] at (4.55, -0.5) {$6$};
	\end{pgfonlayer}
	\begin{pgfonlayer}{edgelayer}
		\draw [style=inarrow] (1) to (4);
		\draw [style=inarrow] (18) to (4);
		\draw [style=inarrow, bend right=15, looseness=1.00] (4) to (10);
		\draw [style=inarrow] (7) to (17);
		\draw [style=inarrow, bend left=15, looseness=1.00] (4) to (10);
		\draw [style=inputarrow] (12) to (1);
		\draw [style=inputarrow] (0) to (18);
		\draw [style=inputarrow] (19) to (18);
		\draw [style=inputarrow] (8) to (17);
		\draw [style=simple] (5.center) to (2.center);
		\draw [style=simple] (2.center) to (15.center);
		\draw [style=simple] (15.center) to (9.center);
		\draw [style=simple] (9.center) to (5.center);
		\draw [style=simple] (6.center) to (13.center);
		\draw [style=simple] (13.center) to (11.center);
		\draw [style=simple] (11.center) to (3.center);
		\draw [style=simple] (3.center) to (6.center);
		\draw [style=inarrow] (10) to (7);
		\draw [style=inarrow]  (7) to (20);
		\draw[ style=inputarrow] (21) to (20);
	\end{pgfonlayer}
\end{tikzpicture}
\]
Note that the states $C$ and $D$ have been identified, and we have arbitrarily called the resulting state $C$.   Thus, $R'R$ has species $S+_Y S' =\{ A,B,C,E,F\}$.   If we had chosen to call the resulting state something else, we would obtain an equivalent open reaction network with rates, and thus the same morphism in $\RxNet$.

At this point we can either compute the vector field $v^{R'R}$ or combine the vector fields $v^R$ and $v^{R'}$ following the procedure given in Theorem \ref{thm:dynam}.  Because $\graysquare$ is a functor, we should get the same answer either way.

The vector field $v^{R'R}$ can be read off from the above picture of $R'R$.  It is
\begin{equation}
\label{eq:v_3}
     v^{R'R}(A,B,C,E,F) = 
\end{equation}
\[
(-r(\alpha) A B, \; -r(\alpha) A B, \; 2 r(\alpha) AB - r(\beta) C, \; r(\beta) C, \; r(\beta) C ) .
\]
On the other hand, the procedure in Theorem \ref{thm:dynam} is to apply the composite of these maps:
\[      
\xymatrix{      D(S) \times D(S') \ar[rr]^-{\delta_{S,S'}} && 
                     D(S + S') \ar[rr]^-{D([j,j'])} && D(S +_Y S') } \]
to the pair of vector fields $(v^R,v^{R'}) \in D(S) \times D(S')$.    The first map
was defined in Lemma \ref{lemma:DLax}, and it yields
\[   \begin{array}{ccl}
  \delta_{S,S'}(v^R,v^{R'}) &=& i_* \circ v^R \circ i^* + i'_* \circ v^{R'} \circ {i'}^* \\  \\
&=&  ( -r(\alpha) AB, \; -r(\alpha) AB , \; 2r(\alpha) AB, \, 0, \; 0, \; 0) \; + \\ 
&& (0, \; 0,\; 0,\, -r(\beta) D , \;r(\beta) D , \; r(\beta) D) \\ \\
&=& ( -r(\alpha) AB,\; -r(\alpha) AB ,\; 2r(\alpha) AB, \; -r(\beta) D , \; r(\beta) D ,\; r(\beta) D).
\end{array}
\]
If we call this vector field $u$, the second map yields
\[  D([j,j'])(u) = [j,j']_* \circ u \circ [j,j']^*.\]
Applying $[j,j']^*$ to any vector of concentrations $(A,B,C,E,F) \in \R^{S +_Y S'}$ yields
$(A,B,C,C,E,F) \in \R^{S+S'}$, since the species $C$ and $D$ are identified by $[j,j']$.   Thus, 
\[   u \circ [j,j']^* =  ( -r(\alpha) AB, -r(\alpha) AB , 2r(\alpha) AB, -r(\beta) C , r(\beta) C , r(\beta) C). \]
Applying $[j,j']_*$ to this, we sum the third and fourth components, again because $C$
and $D$ are identified by $[j,j']$.  Thus,
\[   [j,j']_* \circ u \circ [j,j']^* =  ( -r(\alpha) AB, \, -r(\alpha) AB , \, 2r(\alpha) AB -r(\beta) C , \, r(\beta) C , \, r(\beta) C). \]
As expected, this vector field equals $v^{R'R}$.     The open rate equation of the composite open dynamical system is
\begin{equation}
\label{eq:open_rate_3}
\begin{array}{rcl} 
\displaystyle{\frac{dA(t)}{dt}} &=& - r(\alpha) A(t) B(t)  + I_1(t)\\ \\
\displaystyle{\frac{dB(t)}{dt}} &=& - r(\alpha) A(t) B(t) + I_2(t) + I_3(t) \\ \\
\displaystyle{\frac{dC(t)}{dt}} &=& 2r(\alpha) A(t) B(t) - r(\beta) C(t)  \\ \\
\displaystyle{\frac{dD(t)}{dt}} &=& r(\beta) C(t) - O_5(t) \\ \\
\displaystyle{\frac{dE(t)}{dt}} &=& r(\beta) C(t) - O_6(t) .
\end{array}
\end{equation}

One general lesson here is that when we compose open reaction networks, the process of identifying some of their species via the map $[j,j'] \maps S + S' \to S +_Y S'$ has two effects: copying concentrations and summing reaction velocities.  Concentrations are copied via the pullback $[j,j']^*$, while reaction velocities are summed via the pushfoward $[j,j']_*$.  A similar phenomenon occurs the compositional framework for electrical circuits, where voltages are copied and currents are summed \cite{BaezFong}.  For a deeper look at this, see Section 6.6 of Fong's thesis \cite{FongThesis}.

\section{The black-boxing functor}
\label{sec:black}

The open rate equation describes the behavior of an open dynamical system for any choice of inflows and outflows.  One option is to choose these flows so that the input and output concentrations do not change with time.  In chemistry this is called `chemostatting'.    There will then frequently---though not always---be solutions of the open rate equation where \emph{all} concentrations are constant in time.   These are called `steady states'.  

In this section we take an open dynamical system and extract from it the relation between input and output concentrations and flows that holds in steady state.  We call the process of extracting this relation `black-boxing', since it discards information that cannot be seen at the inputs and ouputs.  The relation thus obtained is always `semi-algebraic', meaning that it can be described by polynomials and inequalities.  In fact, black-boxing defines a functor 
\[        \blacksquare \maps \Dynam \to \SemiAlgRel \]
where $\SemiAlgRel$ is the category of semi-algebraic relations between real vector spaces.    The functoriality of black-boxing means that we can compose two open dynamical systems and then black-box them, or black-box each one and compose the resulting relations: either way, the final answer is the same.   

We can also black-box open reaction networks with rates.  To do this, we simply compose the gray-boxing functor with the black-boxing functor:
\[   \RxNet \stackrel{\graysquare}{\longrightarrow} \Dynam \stackrel{\blacksquare}{\longrightarrow} \SemiAlgRel .\]

We begin by explaining how to black-box an open dynamical system.

\begin{defn}
Given an open dynamical system
$ (X \stackrel{i}\longrightarrow S \stackrel{o}\longleftarrow Y, v) $
we define the \define{boundary} species to be those in $B = i(X) \cup o(Y)$, and the \define{internal} species to be those in $S - B$.
\end{defn}

The open rate equation says
\[ \frac{dc(t)}{dt} = v(c(t)) + i_*(I(t)) - o_*(O(t)) \]
but if we fix $c$, $I$ and $O$ to be constant in time, this reduces to
\[     v(c) + i_*(I) - o_*(O) = 0 .\]   
This leads to the definition of `steady state':

\begin{defn}
Given an open dynamical system 
$ (X \stackrel{i}\longrightarrow S \stackrel{o}\longleftarrow Y, v) $  
together with $I \in \R^X$ and $O \in \R^Y$,
a \define{steady state} with inflows $I$ and outflows $O$ is an element $c \in \R^S$ such that 
\[     v(c) + i_*(I) - o_*(O) = 0 .\]   
\end{defn}

Thus, in a steady state, the inflows and outflows conspire to exactly compensate for
the reaction velocities.    In particular, we must have
\[   \left. v(c)\right|_{S - B} = 0 \]
since the inflows and outflows vanish on internal species.  

\begin{defn}
Given a morphism $F \maps X \to Y$ in $\Dynam$ represented by the open dynamical system 
\[     (X \stackrel{i}\longrightarrow S \stackrel{o}\longleftarrow Y, v), \]
define its \define{black-boxing} to be the set
\[   \blacksquare(F) \subseteq \R^X \oplus \R^X \oplus \R^Y \oplus \R^Y \]
consisting of all 4-tuples $(i^*(c),I,o^*(c),O)$ where $c \in \R^S$ is a steady state
with inflows $I \in \R^X$ and outflows $O \in \R^Y$.
\end{defn}

We call $i^*(c)$ the \define{input concentrations} and $o^*(c)$ the \define{output concentrations}.   Thus, black-boxing records the relation between input concentrations, inflows, output concentrations and outflows that holds in steady state.  This is the `externally observable steady state behavior' of the open dynamical system.

Category theory enters the picture because relations are morphisms in a category.  For any sets $X$ and $Y$, a relation $A \maps X \relto Y$ is a subset $A \subseteq X \times Y$.   Given relations $A \maps X \relto Y$ and $B \maps Y \relto Z$, their composite $B \circ A \maps X \relto Z$ is the set of all pairs $(x,z) \in X \times Z$ such that there exists $y \in Y$ with $(x,y) \in A$ and $(y,z) \in B$.   This gives a category $\Rel$ with sets as objects and relations as morphisms.   

Black-boxing an open dynamical system $F \maps X \to Y$ gives a relation
\[   \blacksquare(F) \maps \R^X \oplus \R^X \relto \R^Y \oplus \R^Y  .\]
This immediately leads to the question of whether black-boxing is a functor from $\Dynam$ to 
$\Rel$.  

The answer is yes.  To get a sense for this, consider the example from Section \ref{sec:gray}, where we composed two open dynamical systems.   We first considered this open reaction network with rates:
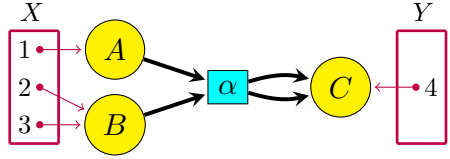
\[
\begin{tikzpicture}
	\begin{pgfonlayer}{nodelayer}
		\node [style=species] (A) at (-4, 0.5) {$A$};
		\node [style=species] (B) at (-4, -0.5) {$B$};
		\node [style=species] (C) at (-1, 0) {$C$};
             \node [style=transition] (a) at (-2.5, 0) {$\alpha$}; 
		
		\node [style=empty] (X) at (-5.1, 1) {$X$};
		\node [style=none] (Xtr) at (-4.75, 0.75) {};
		\node [style=none] (Xbr) at (-4.75, -0.75) {};
		\node [style=none] (Xtl) at (-5.4, 0.75) {};
             \node [style=none] (Xbl) at (-5.4, -0.75) {};
	
		\node [style=inputdot] (1) at (-5, 0.5) {};
		\node [style=empty] at (-5.2, 0.5) {$1$};
		\node [style=inputdot] (2) at (-5, 0) {};
		\node [style=empty] at (-5.2, 0) {$2$};
		\node [style=inputdot] (3) at (-5, -0.5) {};
		\node [style=empty] at (-5.2, -0.5) {$3$};

		\node [style=empty] (Y) at (0.1, 1) {$Y$};
		\node [style=none] (Ytr) at (.4, 0.75) {};
		\node [style=none] (Ytl) at (-.25, 0.75) {};
		\node [style=none] (Ybr) at (.4, -0.75) {};
		\node [style=none] (Ybl) at (-.25, -0.75) {};

		\node [style=inputdot] (4) at (0, 0) {};
		\node [style=empty] at (0.2, 0) {$4$};
		
	\end{pgfonlayer}
	\begin{pgfonlayer}{edgelayer}
		\draw [style=inarrow] (A) to (a);
		\draw [style=inarrow] (B) to (a);
		\draw [style=inarrow, bend left =15] (a) to (C);
		\draw [style=inarrow, bend right =15] (a) to (C);
		\draw [style=inputarrow] (1) to (A);
		\draw [style=inputarrow] (2) to (B);
		\draw [style=inputarrow] (3) to (B);
		\draw [style=inputarrow] (4) to (C);
		\draw [style=simple] (Xtl.center) to (Xtr.center);
		\draw [style=simple] (Xtr.center) to (Xbr.center);
		\draw [style=simple] (Xbr.center) to (Xbl.center);
		\draw [style=simple] (Xbl.center) to (Xtl.center);
		\draw [style=simple] (Ytl.center) to (Ytr.center);
		\draw [style=simple] (Ytr.center) to (Ybr.center);
		\draw [style=simple] (Ybr.center) to (Ybl.center);
		\draw [style=simple] (Ybl.center) to (Ytl.center);
	\end{pgfonlayer}
\end{tikzpicture}
\]
Gray-boxing this gives a morphism in $\Dynam$, say $F \maps X \to Y$, represented by the
open dynamical system
\[         (X \stackrel{i}\longrightarrow S \stackrel{o}\longleftarrow Y, v^R) \]
where the cospan is visible in the figure and $v^R$ is the vector field on $\R^S$ given in
Equation (\ref{eq_v_1}).  If we now black-box $F$, we obtain the relation
\[   \blacksquare(F) = \{  (i^*(c),I,o^*(c),O) : \; v^R(c) + i_*(I) - i_*(O) = 0 \}. \] 
Here the inflows and outflows are
\[   I = (I_1, I_2, I_3) \in \R^X, \qquad O = O_4 \in \R^Y, \]
and vector of concentrations is $c = (A,B,C) \in \R^S$, so the input and output 
concentrations are 
\[   i^*(c) = (A,B,B) \in \R^X, \qquad o^*(c) = C \in \R^Y .\]
To find steady states with inflows $I$ and outflows $O$ we take the open rate equation, Equation (\ref{eq:open_rate_1}), and set all concentrations, inflows and outflows to constants:
\[
\begin{array}{rcr} 
 I_1 &=& r(\alpha) AB\, \\ 
 I_2 + I_3 &=& r(\alpha) AB\, \\ 
 O_4 &=& 2r(\alpha) AB.
\end{array}
\]
Thus, 
\begin{equation}   
\label{eq:black_1}
\blacksquare(F) = 
\end{equation}
\[ \{(A,B,B,I_1,I_2,I_3,C,O_4): \; I_1 = I_2 + I_3 = r(\alpha) AB , 
 O_4 = 2r(\alpha) AB \}.\]

Next we considered this open reaction network with rates:
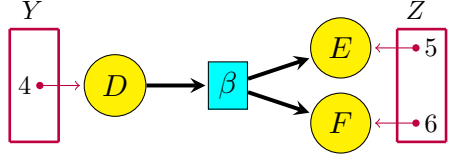
\[
\begin{tikzpicture}
	\begin{pgfonlayer}{nodelayer}
		\node [style = species] (D) at (1, 0) {$D$};
		\node [style = transition] (b) at (2.5, 0) {$\beta$};
		\node [style = species] (E) at (4,0.5) {$E$};
		\node [style = species] (F) at (4,-0.5) {$F$};

		\node [style=empty] (Y) at (-0.1, 1) {$Y$};
		\node [style=none] (Ytr) at (.25, 0.75) {};
		\node [style=none] (Ytl) at (-.4, 0.75) {};
		\node [style=none] (Ybr) at (.25, -0.75) {};
		\node [style=none] (Ybl) at (-.4, -0.75) {};

		\node [style=inputdot] (4) at (0, 0) {};
		\node [style=empty] at (-0.2, 0) {$4$};
		
		\node [style=empty] (Z) at (5, 1) {$Z$};
		\node [style=none] (Ztr) at (4.75, 0.75) {};
		\node [style=none] (Ztl) at (5.4, 0.75) {};
		\node [style=none] (Zbl) at (5.4, -0.75) {};
		\node [style=none] (Zbr) at (4.75, -0.75) {};

		\node [style=inputdot] (5) at (5, 0.5) {};
		\node [style=empty] at (5.2, 0.5) {$5$};	
		\node [style=inputdot] (6) at (5, -0.5) {};
		\node [style=empty] at (5.2, -0.5) {$6$};	

	\end{pgfonlayer}
	\begin{pgfonlayer}{edgelayer}
		\draw [style=inarrow] (D) to (b);
		\draw [style=inarrow] (b) to (E);
		\draw [style=inarrow] (b) to (F);
		\draw [style=inputarrow] (4) to (D);
		\draw [style=inputarrow] (5) to (E);
		\draw [style=inputarrow] (6) to (F);
		\draw [style=simple] (Ytl.center) to (Ytr.center);
		\draw [style=simple] (Ytr.center) to (Ybr.center);
		\draw [style=simple] (Ybr.center) to (Ybl.center);
		\draw [style=simple] (Ybl.center) to (Ytl.center);
		\draw [style=simple] (Ztl.center) to (Ztr.center);
		\draw [style=simple] (Ztr.center) to (Zbr.center);
		\draw [style=simple] (Zbr.center) to (Zbl.center);
		\draw [style=simple] (Zbl.center) to (Ztl.center);
	\end{pgfonlayer}
\end{tikzpicture}
\]
Gray-boxing this gives a morphism $F' \maps Y \to Z$ in $\Dynam$ represented by the open dynamical system 
\[         (Y \stackrel{i'}\longrightarrow S' \stackrel{o'}\longleftarrow Z, v^{R'}) \]
where $v^{R'}$ is given by Equation (\ref{eq:v_2}).   To black-box $F'$ we can follow the same
procedure as for $F$.  We take the open rate equation, Equation (\ref{eq:open_rate_2}), and look
for steady-state solutions:
\[
\begin{array}{rcr} 
I_4 &=& r(\beta) D\,  \\
O_5 &=& r(\beta) D\, \\
O_6 &=& r(\beta) D.
\end{array}
\]
Then we form the relation between input concentrations, inflows, output concentrations and outflows
that holds in steady state:
\begin{equation}
\label{eq:black_2}
\blacksquare(F') = 
 \{ (D,I_4,E,F,O_5,O_6) : \; I_4 = O_5 = O_6 = r(\beta) D \} .  
\end{equation}

Finally, we can compose these two open reaction networks with rates:
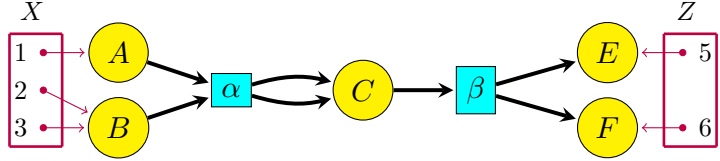
\[
\begin{tikzpicture}
	\begin{pgfonlayer}{nodelayer}
		\node [style=inputdot] (0) at (-4.25, 0) {};
		\node [style=empty] at (-4.55, 0) {$2$};
		\node [style=species] (1) at (-3.25, 0.5) {$A$};
		\node [style=none] (2) at (-4, 0.75) {};
		\node [style=none] (3) at (4, -0.75) {};
		\node [style=transition] (4) at (-1.75, -0) {$\alpha$};
		\node [style=none] (5) at (-4.7, 0.75) {};
		\node [style=none] (6) at (4, 0.75) {};
		\node [style=transition] (7) at (1.5, -0) {$\beta$};
		\node [style=inputdot] (8) at (4.25, 0.5) {};
		\node [style=empty] at (4.55, 0.5) {$5$};
		\node [style=none] (9) at (-4.7, -0.75) {};
		\node [style=species] (10) at (0, -0) {$C$};
		\node [style=none] (11) at (4.7, -0.75) {};
		\node [style=inputdot] (12) at (-4.25, 0.5) {};
		\node [style=empty] at (-4.55, 0.5) {$1$};
		\node [style=none] (13) at (4.7, 0.75) {};
		\node [style=empty] (14) at (-4.4, 1.05) {$X$};
		\node [style=none] (15) at (-4, -0.75) {};
		\node [style=empty] (16) at (4.3, 1.05) {$Z$};
		\node [style=species] (17) at (3.25, 0.5) {$E$};
		\node [style=species] (18) at (-3.25, -0.5) {$B$};
		\node [style=inputdot] (19) at (-4.25, -0.5) {};
		\node [style=empty] at (-4.55, -0.5) {$3$};
		\node[ style=species] (20) at (3.25, -0.5) {$F$};
		\node[ style=inputdot] (21) at (4.25, -0.5) {};
		\node [style=empty] at (4.55, -0.5) {$6$};
	\end{pgfonlayer}
	\begin{pgfonlayer}{edgelayer}
		\draw [style=inarrow] (1) to (4);
		\draw [style=inarrow] (18) to (4);
		\draw [style=inarrow, bend right=15, looseness=1.00] (4) to (10);
		\draw [style=inarrow] (7) to (17);
		\draw [style=inarrow, bend left=15, looseness=1.00] (4) to (10);
		\draw [style=inputarrow] (12) to (1);
		\draw [style=inputarrow] (0) to (18);
		\draw [style=inputarrow] (19) to (18);
		\draw [style=inputarrow] (8) to (17);
		\draw [style=simple] (5.center) to (2.center);
		\draw [style=simple] (2.center) to (15.center);
		\draw [style=simple] (15.center) to (9.center);
		\draw [style=simple] (9.center) to (5.center);
		\draw [style=simple] (6.center) to (13.center);
		\draw [style=simple] (13.center) to (11.center);
		\draw [style=simple] (11.center) to (3.center);
		\draw [style=simple] (3.center) to (6.center);
		\draw [style=inarrow] (10) to (7);
		\draw [style=inarrow]  (7) to (20);
		\draw[ style=inputarrow] (21) to (20);
	\end{pgfonlayer}
\end{tikzpicture}
\]
Gray-boxing the composite gives a morphism $F'F \maps X \to Z$ represented by the open dynamical system
\[         (Y \stackrel{i'}\longrightarrow S' \stackrel{o'}\longleftarrow Z, v^{R'R}) \]
where $v^{R'R}$ is given in Equation (\ref{eq:v_3}).   To black-box $F'F$ we take its open rate equation, Equation (\ref{eq:open_rate_3}), and look for steady state solutions:
\[
\begin{array}{rcl} 
I_1 &=& r(\alpha) AB \\
I_2 + I_3 &=& r(\alpha) AB \\ 
r(\beta) C &=& 2r(\alpha) AB  \\
O_5 &=& r(\beta) C \\
O_6 &=& r(\beta) C .
\end{array}
\]
The concentrations of internal species play only an indirect role after we black-box an open
dynamical system, since black-boxing only tells us the steady state relation between input concentrations, inflows, output concentrations and outflows.  In $F$ and $F'$ there were no internal species.  In $F'F$ there is one, namely $C$.  However, in this particular example the concentration $C$ is completely determined by the other data, so we can eliminate it from the above equations.  This is not true in every example.  But we can take advantage of this special feature here, obtaining these equations:
\[
\begin{array}{rcl} 
I_1 &=& r(\alpha) AB \\
I_2 + I_3 &=& r(\alpha) AB \\ 
O_5 &=& 2r(\alpha) AB  \\
O_5 &=& O_6 .
\end{array}
\]
We thus obtain
\begin{equation}
\label{eq:black_3}
\blacksquare(F'F)= 
\end{equation}
\[  \{ (A,B,B,I_1,I_2,I_3,E,F,O_5,O_6) : \; I_1  = I_2 + I_3 = r(\alpha) AB, O_5 = O_6 = 2 r(\alpha) AB \}   .\]
We leave it to the reader to finish checking the functoriality of black-boxing in this example:
\[         \blacksquare(F' F) = \blacksquare(F') \blacksquare(F) .\]
To do this, it suffices to compose the relations $\blacksquare(F)$ given in Equation (\ref{eq:black_1}) and $\blacksquare(F')$ given in Equation (\ref{eq:black_2}).

This example was a bit degenerate, because in each open dynamical system considered there was at most one steady state compatible with any choice of input concentrations, inflows, output concentrations and outflows.   In other words, even when there was an internal species, its concentration was determined by this `boundary data'.    This is far from generally true! 
Even for relatively simple `closed' reaction networks, namely those with no boundary species, multiple steady states may be possible.  Such reaction networks often involve features such as `autocatalysis', meaning that a certain species is present as both an input and an output to the same reaction.   Criteria for the uniqueness of steady states have been a major focus of reaction network theory ever since its birth \cite{BanajiCraciun, Feinberg, FeinbergHorn, Feinberg1995a}, and so have criteria for the existence of multiple steady states \cite{CraciunFeinbergTang, Feinberg1995b}.  We expect the study of open reaction networks to give a new outlook on these questions.  However, our proof of the functoriality of black-boxing sidesteps this issue.

Before proving this result, it is nice to refine the framework slightly.  The black-boxing of an open dynamical system is far from an arbitrary relation: it is always `semialgebraic'. To understand this, we need a lightning review of semialgebraic geometry \cite{Coste}.

Let us use `vector space' to mean a finite-dimensional real vector space.  Given a vector space $V$, the collection of \define{semialgebraic subsets} of $V$ is the smallest collection that contains all sets of the form $\{P(v) = 0\}$ and $\{P(v) > 0\}$, where $P \maps V \to \R$ is any polynomial, and is closed under finite intersections, finite unions and complements.   The Tarski--Seidenberg theorem says that if $S \subseteq V \oplus W$ is semialgebraic then so is its projection to $V$, that is, the subset of $V$ given by
\[                \{v \in V :\; \exists w \in W \; (v,w) \in S \} .\]

If $U$ and $V$ are vector spaces, a \define{semialgebraic relation} $A \maps U \relto V$ is a semialgebraic subset $A \subseteq U \oplus V$.    If $A \maps U \relto V$ and $B \maps V \relto W$ are semialgebraic relations, so is their composite
\[      B \circ A = \{(u,w) : \; \exists v \in V \; (u,v) \in A \textrm{ and } (v,w) \in B \} \]
thanks to the Tarski--Seidenberg theorem.   The identity relation on any vector space is
also semialgebraic, so we obtain a category:

\begin{defn}
Let $\SemiAlgRel$ be the category with vector spaces as objects and semialgebraic relations as morphisms.
\end{defn}

This category is symmetric monoidal in a natural way, where the tensor product of objects is given by the direct sum of vector spaces.   We can now state the main theorem about black-boxing:

\begin{thm}
\label{thm:black}
There is a symmetric monoidal functor $\blacksquare \maps \Dynam \to \SemiAlgRel$
sending any finite set $X$ to the vector space $\R^X \oplus \R^X$ and any morphism $F \maps X \to Y$ to its black-boxing $\blacksquare(F)$.
\end{thm}

\begin{proof}
For any morphism $F \maps X \to Y$ in $\Dynam$ represented by the open dynamical system 
\[         (X \stackrel{i}\longrightarrow S \stackrel{o}\longleftarrow Y, v) \]
the set 
\[       \{(c, i^*(c),I,o^*(c),O) : \; v(c) + i_*(I) - i_*(O) = 0 \} \subseteq
\R^S \oplus \R^X \oplus \R^X \oplus \R^Y \oplus \R^Y \] 
is defined by polynomial equations, since $v$ is algebraic.   Thus, by the Tarski--Seidenberg theorem, the set 
\[    \blacksquare(F) = \{  (i^*(c),I,o^*(c),O) : \; v(c) + i_*(I) - i_*(O) = 0 \} \]
is semialgebraic.   

Next we prove that $\blacksquare$ is a functor.  Consider composable morphisms $F \maps X \to Y$ and $F' \maps Y \to Z$ in $\Dynam$.   We know that $F$ is represented by some open dynamical system 
\[         (X \stackrel{i}\longrightarrow S \stackrel{o}\longleftarrow Y, v) \]
while $F'$ is represented by some
\[         (Y \stackrel{i'}\longrightarrow S' \stackrel{o'}\longleftarrow Z, v') .\]
To compose these, we form the pushout
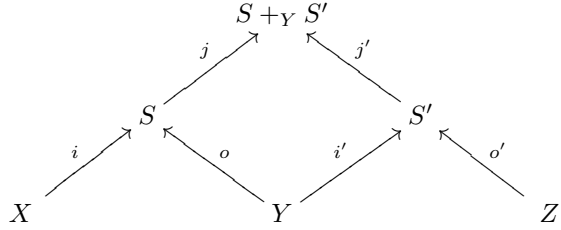
\[
    \xymatrix{
      && S +_Y S' \\
      & S \ar[ur]^{j} && S' \ar[ul]_{j'} \\
      \quad X\quad \ar[ur]^{i} && Y \ar[ul]_{o} \ar[ur]^{i'} &&\quad Z \quad \ar[ul]_{o'}
    }
\]
Then $F'F \maps X \to Z$ is represented by the open dynamical system
\[ (X \stackrel{j i}{\longrightarrow} S +_Y S' \stackrel{j' o'}{\longrightarrow} Z, u ) \]
where 
\[    u = j_* \circ v \circ j^* + {j'}_* \circ v' \circ {j'}^*  .\]

To prove that $\blacksquare$ is a functor, we first show that 
\[  \blacksquare(F') \blacksquare(F) \subseteq \blacksquare(F'F). \]  
Thus, given
\[     (i^*(c),I,o^*(c),O) \in \blacksquare(F), \qquad  ({i'}^*(c'),I',{o'}^*(c'),O') \in \blacksquare(F') \]
with 
\[   o^*(c) = {i'}^*(c'), \qquad O = I' \]
we need to prove that 
\[     (i^*(c),I,{o'}^*(c'),O') \in \blacksquare(F'F). \]
To do this, it suffices to find concentrations $b \in \R^{S +_Y S'}$ such that 
\[   (i^*(c),I,{o'}^*(c'),O') = ((ji)^*(b), I, {(j'o')}^*(b), O') \]
and $b$ is a steady state of $F'F$ with inflows $I$ and outflows $O'$.

Since  $o^*(c) = {i'}^*(c'),$ this diagram commutes:
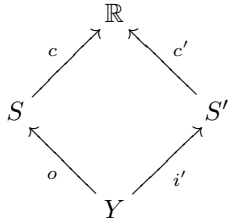
\[
    \xymatrix{
      && \R \\
      & S \ar[ur]^{c} && S' \ar[ul]_{c'} \\
       && Y \ar[ul]^{o} \ar[ur]_{i'} &&
    }
\]
so by the universal property of the pushout there is a unique map $b \maps S +_Y S' \to \R$ such that
this commutes:
\begin{equation}
\label{eq:pushout}
    \xymatrix{
      && \R \\
     && S +_Y  S' \ar[u]^b \\      
      & S \ar@/^/[uur]^{c} \ar[ur]^{j} && S' \ar@/_/[uul]_{c'} \ar[ul]_{j'} \\
       && Y \ar[ul]^{o} \ar[ur]_{i'} &&
    }
\end{equation}
This simply says that because the concentrations $c$ and $c'$ agree on the `overlap' of our two
open dynamical systems, we can find a concentration $b$ for the composite system that restricts
to $c$ on $S$ and $c'$ on $S'$.

We now prove that $b$ is a steady state of the composite open dynamical system with
inflows $I$ and outflows $O'$:
\begin{equation}
\label{eq:steady_state_3}
   u(b) + (ji)_*(I) - (j'o')_*(O') = 0.
\end{equation}
To do this we use the fact that $c$ is a steady state of $F$ with inflows $I$ and outflows $O$:
\begin{equation}
\label{eq:steady_state_1}
   v(c) + i_*(I) - {o}_*(O) = 0
\end{equation}
and $c'$ is a steady state of $F'$ with inflows $I'$ and outflows $O'$:
\begin{equation}
\label{eq:steady_state_2}
   v'(c') + {i'}_*(I') - {o'}_*(O') = 0.
\end{equation}
We push forward Equation (\ref{eq:steady_state_1}) along $j$, push forward Equation
(\ref{eq:steady_state_2}) along $j'$, and sum them:
\[   j_*(v(c))  + (ji)_*(I) - (jo)_*(O) + j'_*(v'(c')) + (j'i')_*(I') - (j'o')_*(O') = 0. \]
Since $O = I'$ and $jo = j'i'$, two terms cancel, leaving us with
\[     j_*(v(c))  + (ji)_*(I) + j'_*(v'(c')) - (j'o')_*(O') = 0. \]
Next we combine the terms involving the vector fields $v$ and $v'$, with the help of Equation (\ref{eq:pushout}) and the definition of $u$:
\begin{equation}
\label{eq:u}
   \begin{array}{ccl}
  j_*(v(c)) + j'_*(v'(c')) &=& j_*(v(b \circ j)) + j'_*(v'(b \circ j')) \\
                                    &=& (j_* \circ v \circ j^* + j'_* \circ v' \circ j'^*)(b) \\
                                    &=& u(b)  .
\end{array}
\end{equation}
This leaves us with
\[         u(b) +  (ji)_*(I) - (j'o')_*(O') = 0 \]
which is Equation (\ref{eq:steady_state_3}), precisely what we needed to show.

To finish showing that $\blacksquare$ is a functor, we need to show that 
\[   \blacksquare(F'F) \subseteq \blacksquare(F') \blacksquare(F)  .\] 
So, suppose we have 
\[    ((ji)^*(b), I, {(j'o')}^*(b), O') \in \blacksquare(F'F) .\]
We need to show
\begin{equation}
\label{eq:composite}
  ((ji)^*(b), I, {(j'o')}^*(b), O') = (i^*(c),I,{o'}^*(c'),O')
\end{equation}
where 
\[     (i^*(c),I,o^*(c),O) \in \blacksquare(F), \qquad  ({i'}^*(c'),I',{o'}^*(c'),O') \in \blacksquare(F') \]
and
\[   o^*(c) = {i'}^*(c'), \qquad O = I' .\]

To do this, we begin by choosing
\[   c = j^*(b), \qquad c' = {j'}^*(b) .\]
This ensures that Equation (\ref{eq:composite}) holds, and since $jo = j'i'$, it also ensures that 
\[  o^*(c) = (jo)^*(b) = (j'i')^*(b) = {i'}^*(c')  .\]
So, to finish the job, we only need to find an element $O = I' \in \R^Y$ such that $c$ is a steady state of $F$ with inflows $I$ and outflows $O$ and $c'$ is a steady state of $F'$ with inflows $I'$ and outflows $O'$.  Of course, we are given the fact that $b$ is a steady state of $F'F$ with inflows $I$ and outflows $O'$.   

In short, we are given Equation (\ref{eq:steady_state_3}), and we want to find $O = I'$ such that Equations (\ref{eq:steady_state_1}) and (\ref{eq:steady_state_2}) hold.  Thanks to our choices of $c$ and $c'$,  we can use Equation (\ref{eq:u}) and rewrite Equation (\ref{eq:steady_state_3}) as
\begin{equation}
\label{eq:steady_state_3'}
  j_*(v(c) + i_*(I)) \; + \; {j'}_*(v'(c') - {o'}_*(O')) = 0 .  
\end{equation}
Equations  (\ref{eq:steady_state_1}) and (\ref{eq:steady_state_2}) say that
\begin{equation}
\label{eq:steady_state_1'2'}
\begin{array}{lcl}
   v(c) + i_*(I) - {o}_*(O) &=& 0 \\  \\
   v'(c') + {i'}_*(I') - {o'}_*(O') &=& 0.
\end{array}
\end{equation}

Now we use the fact that 
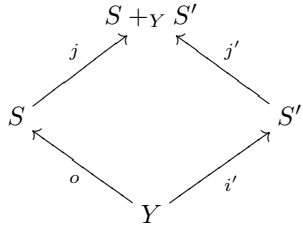
\[
    \xymatrix{
      & S +_Y S' \\
       S \ar[ur]^{j} && S' \ar[ul]_{j'} \\
       & Y \ar[ul]^{o} \ar[ur]_{i'} &
    }
\]
is a pushout.  Applying the `free vector space on a finite set' functor, which preserves colimits, this implies that
\[
    \xymatrix{
      & \R^{S +_Y S'} \\
       \R^S \ar[ur]^{j_*} && \R^{S'} \ar[ul]_{{j'}_*} \\
       & \R^Y \ar[ul]^{o_*} \ar[ur]_{i'_*} &
    }
\]
is a pushout in the category of vector spaces.   Since a pushout is formed by taking first a coproduct and then a coequalizer, this implies that 
\[
     \xymatrix{
      \R^Y \ar@<-.5ex>[rr]_-{(0,i'_*)} \ar@<.5ex>[rr]^-{(o_*,0)} && \R^S \oplus \R^{S'} \ar[rr]^{j_* + j'_*}
   && \R^{S +_Y S'}
}
\]
is a coequalizer.  Thus, the kernel of $j_* + j'_*$ is the image of $(o_*,0) - (0,i'_*)$.   Equation (\ref{eq:steady_state_3'}) says precisely that 
\[    (v(c) + i_*(I), v'(c') - o'_*(O')) \in \ker(j_* + j'_*)  .\]
Thus, it is in the image of $o_* - i'_*$.  In other words, there exists some element $O = I' \in \R^Y$
such that 
\[   (v(c) + i_*(I), v'(c') - o'_*(O')) = (o_*(O), -i'_*(I')).\]
This says that Equations (\ref{eq:steady_state_1}) and (\ref{eq:steady_state_2}) hold, as desired.

Finally, we need to check that $\blacksquare$ is symmetric monoidal.  But this is a straightforward calculation, so we leave it to the reader.
\end{proof}

This theorem can also be proved using Fong's theory of `decorated corelations' \cite{FongThesis,Fong2017}.  Decorated corelations are generalization of decorated cospans well-suited to extracting the externally observable behavior from an open system.    The category 
$\SemiAlgRel$ is an example of a decorated corelation category that is not a decorated cospan category.

\section{Conclusions}
\label{sec:conclusions}

It is worth comparing our black-boxing theorem, Theorem \ref{thm:black}, to Spivak's work on open dynamical systems \cite{Spivak}.  He describes various  categories where the morphisms are open dynamical systems, and constructs functors from these categories to $\Rel$, which describe the steady state relations between inputs and outputs.  None of his results subsume ours, but they are philosophically very close.  Both are doubtless special cases of a more general theorem that is yet to be formulated.

It is also worth comparing Theorem \ref{thm:black} to our previous black-boxing theorem for open Markov processes \cite{BaezFongPollard}.  In this previous work, we began by describing a category $\Mark$ whose morphisms are equivalence classes of open Markov processes.  In fact $\Mark$ is equivalent to the subcategory of $\RxNet$ with all the same objects but only those morphisms coming from open reaction networks for which each transition has exactly one input and one output.  There is
thus an inclusion of categories
\[        I \maps \Mark \to \RxNet .\]

We did not construct a black-boxing functor for $\Mark$, but only for a related category $\DetBalMark$ where the morphisms are equivalence classes of `detailed balanced' Markov process.  Since the differential equation describing a Markov process is linear, we obtained a black-boxing functor
\[  \square \maps \DetBalMark \to \LinRel  \]
where $\LinRel$ is the category of vector spaces and \emph{linear} relations.    

Since every open detailed balanced Markov process has an underlying open Markov process, there is a forgetful functor 
\[     F \maps \DetBalMark \to \Mark . \]
Since every linear relation is a semialgebraic relation, there is also a functor
\[     U \maps \LinRel \to \SemiAlgRel . \]
Combining all the results mentioned so far, we obtain this diagram:
\[ 
\xymatrix{
\DetBalMark \ar[r]^-F \ar[ddr]_{\square} & \Mark \ar[r]^I & \RxNet  \ar[d]^{\graysquare} \\ 
& &   \Dynam \ar[d]^{\blacksquare} \\
& \LinRel \ar[r]_-U & \SemiAlgRel 
}
\]
This commutes up to natural isomorphism, because our previous black-boxing functor $\square$ agrees with our new procedure in the special case of open detailed balanced Markov processes.  Furthermore, our new procedure assigns a \emph{linear} relation to any open dynamical system coming from an open Markov processes.  Thus, there is a functor 
\[           \darkgraysquare \maps \Mark \to \LinRel  \]
making this diagram commute up to natural isomorphism:
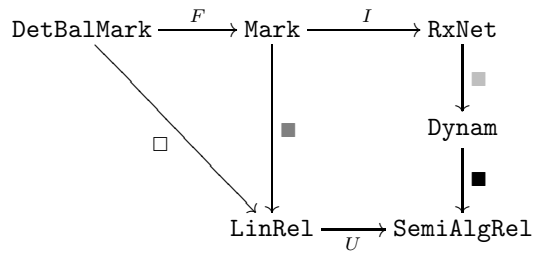
\[ 
\xymatrix{
\DetBalMark \ar[r]^-F \ar[ddr]_{\square} & \Mark \ar[r]^I  \ar[dd]^{\darkgraysquare} & \RxNet  \ar[d]^{\graysquare} \\ 
& &   \Dynam \ar[d]^{\blacksquare} \\
& \LinRel \ar[r]_-U & \SemiAlgRel 
}
\]
Our previous work obtained the functor $\square \maps \DetBalMark \to \LinRel$ by reducing
detailed balanced Markov processes to electrical circuits made of resistors and using an already established procedure for black-boxing these \cite{BaezFong}.   We constructed a triangle of functors that commutes up to natural isomorphism:
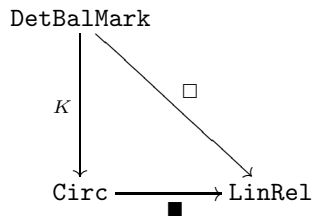
\[ 
\xymatrix{
\DetBalMark \ar[dd]_K \ar[ddr]^{\square}  \\ 
& &   \\
\Circ \ar[r]_{\blacksquare} & \LinRel 
}
\]
Here $\Circ$ is a category whose morphisms are equivalence classes of open circuits made of
resistors, and $\vdarkgraysquare \maps \Circ \to \LinRel$ is the black-boxing functor for these.

Combining all these results, we obtain the following diagram of functors:
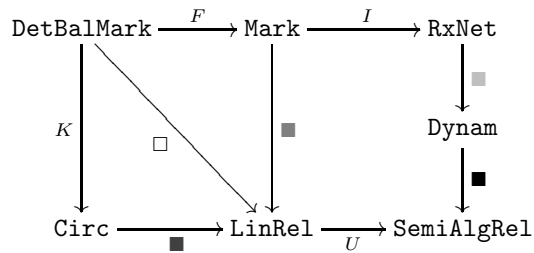
\[ 
\xymatrix{
\DetBalMark \ar[dd]_K \ar[r]^-F \ar[ddr]_{\square} & \Mark \ar[r]^I  \ar[dd]^{\darkgraysquare} & \RxNet  \ar[d]^{\graysquare} \\ 
& &   \Dynam \ar[d]^{\blacksquare} \\
\Circ \ar[r]_{\vdarkgraysquare} & \LinRel \ar[r]_-U & \SemiAlgRel 
}
\]
This diagram, which commutes up to natural isomorphism, makes our current understanding of these various subjects seem a bit more unified than it actually is.  Steady states of circuits made of resistors obey a variational principle: the principle of least power \cite{BaezFong}.  Similarly, steady states of detailed balanced Markov processes obey the principle of minimum dissipation \cite{BaezFongPollard}.  We used these variational principles to define the black-boxing functors for these systems, and as a result we obtained not merely linear relations between vector spaces, but \emph{Lagrangian} relations between \emph{symplectic} vector spaces.  As far as we know, the steady states of more general Markov processes, or reaction networks with rates, do not obey any extremal principle. 

\appendix
\section{Decorated cospan categories}
\label{sec:deccospan}

For ease of reference, we recall some definitions and theorems from Fong's work on decorated cospans \cite{Fong2015,FongThesis}.  We assume some familiarity with basic category theory, such as colimits and symmetric monoidal categories \cite{MacLane}.
We need the concept of lax symmetric monoidal functor:

\begin{defn} \label{defn.lmf}
  Let $(\CC,\boxtimes)$ and $(\D,\otimes)$ be symmetric monoidal categories. A \define{lax symmetric monoidal functor} 
  \[
    (F,\varphi) \maps (\CC,\boxtimes) \to (\D,\otimes)
  \]
  is a functor $F \maps \CC \to \D$ equipped with a natural transformation 
  \[
    \varphi_{-,-} \maps F(-)\otimes F(-) \Rightarrow F(-\boxtimes -)
  \]
and a morphism
  \[
    \varphi \maps I_\D \to F(I_{\CC})
  \]
where $I_\CC$ and $I_\D$ are the unit objects for $\CC$ and $\D$, 
  such that four diagrams commute.  These diagrams are the \define{associator hexagon}:
 \[
   \xymatrix{ F(A) \otimes (F(B) \otimes F(C)) \ar[d]_{\mathrm{id} \otimes \varphi_{B,C}} 
                    \ar[rr]^{\sim} &&
                    (F(A) \otimes F(B)) \otimes F(C) \ar[d]^{\varphi_{A,B}
		    \otimes \mathrm{id}}  \\
                    F(A) \otimes F(B \boxtimes C) \ar[d]_{\varphi_{A,B\boxtimes C}} &&
                     F(A \boxtimes B) \otimes F(C) \ar[d]^{\varphi_{A\boxtimes B,C}} \\
                     F(A \boxtimes (B\boxtimes C)) \ar[rr]^{\sim} && 
                     F((A \boxtimes B) \boxtimes C)                    
    }
  \]
where the horizontal arrows come from the associators for $\otimes$ and
$\boxtimes$, the \define{left and right unitor squares}:
\[ 
  \xymatrix{
    1_\D \otimes F(A) \ar[r]^-{\varphi \otimes \mathrm{id}} \ar[d]_{\sim} &
    F(1_{\CC}) \otimes F(A) \ar[d]^{\varphi_{1_\CC,A}} \\
    F(A) & F(1_\CC \boxtimes A) \ar[l]_{\sim}
  }
\]
and
\[ 
  \xymatrix{
    F(A) \otimes 1_\D \ar[r]^-{\mathrm{id} \otimes \varphi} \ar[d]_{\sim} &
    F(A) \otimes F(1_\CC) \ar[d]^{\varphi_{A,1_\CC}} \\
    F(A) & F(A \boxtimes 1_\CC) \ar[l]_{\sim}
  }
\]
where the isomorphisms come from the unitors for $\otimes$ and $\boxtimes$, and
the \define{braiding square}:
\[ 
  \xymatrix{
    F(A) \boxtimes F(B) \ar[r]^-{\varphi_{A,B}} \ar[d]_{\sim} &
    F(A \otimes B) \ar[d]^{\sim} \\
    F(B) \boxtimes F(A) \ar[r]^-{\varphi_{B,A}} & F(B \otimes A) 
  }
\]
where the isomorphisms come from the braidings for $\otimes$ and $\boxtimes$.
\end{defn}

Let $(\Set, \times)$ denote the category of sets and functors made into a symmetric monoidal category with the cartesian product as its tensor product.  We write $1$ for a chosen one-element set serving as the unit for this tensor product.  Given a category $\CC$ with finite colimits, we let $(\CC,+)$ denote this category made into a symmetric monoidal category with the coproduct as its tensor product, and write $0$ for a chosen initial object serving as the unit of this tensor product.

The decorated cospan construction is then as follows.  Suppose $\CC$ is a category with finite colimits and 
  \[
    (F,\varphi)\maps (\CC,+) \longrightarrow (\Set, \times)
  \]
   is a lax symmetric monoidal functor.  

\begin{defn}
   Given objects $X, Y \in \CC$, an \define{$F$-decorated cospan} from $X$ to $Y$ 
  is a pair consisting of a cospan $X \stackrel{i}\longrightarrow S \stackrel{o}\longleftarrow Y$ in $\CC$ and an element $d \in F(S)$.  We call $d$ the \define{decoration}.
\end{defn}

\begin{defn}
Given two decorated cospans
 \[
    (X \stackrel{i}\longrightarrow S \stackrel{o}\longleftarrow Y, d) 
    \quad \textrm{ and } \quad
    (Y \stackrel{i'}\longrightarrow S' \stackrel{o'}\longleftarrow Z, d'), 
  \]
their \define{composite} is the cospan from $X$ to $Y$ constructed via a pushout as follows:
  \[
    \xymatrix{
      && S +_Y S' \\
      & S \ar[ur]^{j} && S' \ar[ul]_{j'} \\
      \quad X\quad \ar[ur]^{i} && Y \ar[ul]_{o} \ar[ur]^{i'} &&\quad Z \quad \ar[ul]_{o'}
    }
  \]
together with the decoration obtained by applying the map
\[      
\xymatrix{      F(S) \times F(S') \ar[rr]^-{\varphi_{S,S'}} && 
                     F(S + S') \ar[rr]^-{F([j,j'])} && F(S +_Y S') } \]
to the pair $(d,d') \in F(S) \times F(S')$.  Here $[j,j'] \maps S + S' \to S +_Y S'$ is the
canonical morphism from the coproduct to the pushout.
\end{defn}

Composition of decorated cospans is not associative until we quotient by the following 
equivalence relation.   It is straightforward to check that composition of decorated cospans is well-defined at the level of equivalence classes.

\begin{defn}
Two decorated
cospans from $X \in \CC$ to $Y \in \CC$
\[
   (X \stackrel{i}\longrightarrow S \stackrel{o}\longleftarrow Y, d) 
 \quad \textrm{ and } \quad
 (X \stackrel{i'}\longrightarrow S' \stackrel{o'}\longleftarrow Y, d')
  \]
are said to be \define{equivalent} if there is an isomorphism $f \maps S \to S'$ such that $F(f)(d) = d'$ and this diagram commutes:
\[
  \xymatrix{
    & S \ar[dd]^n  \\
    X \ar[ur]^{i} \ar[dr]_{i'} && Y \ar[ul]_{o} \ar[dl]^{o'}\\
    & S'.
  }
\]
\end{defn}

\begin{lem}[\textbf{Fong}] \label{lemma:fcospans}
There is a category $F\Cospan$, whose objects are those of $\CC$ and whose
morphisms from $X \in \CC$ to $Y \in \CC$ are equivalence classes of decorated cospans from $X$ to $Y$, composed as above.  
\end{lem}

\begin{proof}
This is Proposition 3.2 in Fong's paper on decorated cospans \cite{Fong2015}.
\end{proof}

In fact $F\Cospan$ is much better than a mere category.  To begin with, it is
a monoidal category, where the tensor product of objects is the disjoint union of sets and the tensor product of morphisms arising from two decorated cospans 
\[
    (X \stackrel{i}{\longrightarrow} S \stackrel{o}{\longleftarrow} Y, d) 
    \quad \textrm{ and } \quad
    (X' \stackrel{i'}{\longrightarrow} S' \stackrel{o'}{\longleftarrow} Y', d')
  \]
is the morphism associated to this decorated cospan:
\[  ( X + X' \stackrel{i+i'}{\longrightarrow} S + S'  \stackrel{o + o'}{\longleftarrow} Y + Y', \;
\varphi_{S,S'}(d,d') ) .\]
Fong goes further: he proves that $F\Cospan$ is a specially nice sort of 
symmetric monoidal category, called a `dagger compact category',
in which any morphism $f \maps A \to B$ can be turned around to give a 
morphism $f^\dagger \maps B \to A$.  The dagger of 
\[
    (X \stackrel{i}{\longrightarrow} S \stackrel{o}{\longleftarrow} Y, d) 
\]
is simply 
\[
    (Y \stackrel{o}{\longrightarrow} S \stackrel{i}{\longleftarrow} X, d) 
\]
but some properties must hold for this operation to give a dagger compact category.

The importance of such categories was highlighted in Abramsky, Coecke and 
Selinger's work on quantum theory \cite{AC,Selinger}.   But in fact, 
$F\Cospan$ is a particularly well-behaved sort of dagger compact category: a 
`hypergraph category'.   These were first introduced by Carboni \cite{Carboni} under another name, `well-supported compact closed categories'.  They were renamed `hypergraph categories' by Kissinger \cite{Kissinger}, in reference to the fact that such categories contain morphisms corresponding to arbitrary hypergraphs: that is, diagrams with wires connecting any number of inputs to any number of outputs.   The definition and significance of all these concepts are clearly explained in Fong's thesis \cite{FongThesis}.  

\begin{thm}[\textbf{Fong}] \label{thm:fcospans}
  Suppose $\CC$ is a category with finite colimits and 
  \[
    (F,\varphi)\maps (\CC,+) \longrightarrow (\Set, \times)
  \]
is a lax symmetric monoidal functor.  Then the category $F\Cospan$ is a symmetric
monoidal category, and in fact a hypergraph
category.
\end{thm}

\begin{proof}  
This is Theorem 3.4 in Fong's paper on decorated cospans \cite{Fong2015}.
\end{proof}

Decorated cospan categories help us understand diagrams of open systems, and
how to manipulate them.  We also need the ability to take such diagrams and read off 
their `meaning', for example via the rate equation.  For this we need functors between
decorated cospan categories.   Fong found a way to construct these from monoidal natural transformations \cite{Fong2015}.   Moreover, he proved that the resulting functors are actually  `hypergraph functors', meaning the sort that preserves all the structure of a hypergraph category.  As a corollary they are a symmetric monoidal dagger functors.

\begin{defn}
\label{def:monnattran}
  A \define{monoidal natural transformation} $\alpha$ from a lax symmetric monoidal
  functor
  \[    (F,\varphi)\maps (\CC,\otimes) \longrightarrow (\Set,\times)
  \]
  to a lax symmetric monoidal functor
  \[
    (G,\gamma)\maps (\CC,\otimes) \longrightarrow (\Set,\times)
  \]
  is a natural transformation $\alpha\maps F \Rightarrow G$ such that
  \[
    \xymatrix{
      F(A) \times F(B) \ar[r]^-{\varphi_{A,B}} \ar[d]_{\alpha_A \times
      \alpha_B} & F(A \otimes B) \ar[d]^{\alpha_{A\otimes B}} \\
      G(A) \times G(B) \ar[r]^-{\gamma_{A,B}} & G(A \otimes B)
    }
  \]
  commutes.
\end{defn}

\begin{thm}[\textbf{Fong}] \label{thm:decoratedfunctors}
  Let $\CC$ be a category with finite colimits and let
  \[
    (F,\varphi) \maps (\CC,+) \longrightarrow (\Set,\times)
  \]
  and
  \[
    (G,\gamma) \maps (\CC,+) \longrightarrow (\Set,\times)
  \]
  be lax symmetric monoidal functors. This gives rise to decorated cospan
  categories $F\Cospan$ and $G\Cospan$. 

  Suppose furthermore that we have a monoidal natural transformation $\theta_S \maps F(s) \to G(s)$.   Then there is a symmetric monoidal functor 
  \[
    T \maps F\Cospan \longrightarrow G\Cospan
  \]
  mapping each object $X \in F\Cospan$ to the object $X \in G\Cospan$ and
  each morphism represented by an $F$-decorated cospan
  \[
    (X \stackrel{i}\longrightarrow S \stackrel{o}\longleftarrow Y, d)
  \]
  to the morphism represented by the $G$-decorated cospan
  \[
    (X \stackrel{i}{\longrightarrow} S \stackrel{o}{\longleftarrow} Y,
    \; 
    \theta_S(d) ).
  \]
In fact $T$ is a hypergraph functor.
\end{thm}

\begin{proof}
This is a special case of Theorem 4.1 in Fong's paper on decorated cospans \cite{Fong2015}.  
\end{proof}

\end{document}